\documentclass[letterpaper,twocolumn,10pt]{article}
\usepackage{usenix-2020-09}

\usepackage{tikz}
\usepackage{amsmath}

\usepackage{xspace}
\usepackage{multirow}
\usepackage{subfig}

\usepackage{algorithmicx}
\usepackage[ruled,vlined]{algorithm2e} 
\usepackage{algpseudocode}
\usepackage{algcompatible}

\usepackage{algpascal}
\usepackage{algc}
\usepackage{amsthm}
\usepackage{amsfonts}

\usepackage{mathrsfs}
\usepackage{todonotes}
\usepackage{url}

\usepackage{enumitem}

\usepackage{thmtools}

\usepackage{authblk}

\makeatletter
\def\thm@space@setup{%
  \thm@preskip=1pt
  \thm@postskip=\thm@preskip
}
\makeatother

\usepackage{tablefootnote}

\newcommand{\sys}{\textsf{FlexTender}\xspace}

\newcommand{\msgtag}[1]{\textsc{{#1}}}

\newcommand{\msg}[2]{\ensuremath{\langle\msgtag{#1},#2\rangle}}

\newcommand{\reviewed}{examined\xspace}

\newcommand{\eov}{\textsc{eov-s}im\xspace}

\newcommand{\propose}{\textsc{propose}}
\newcommand{\proposal}{\textsc{proposal}}
\newcommand{\prevote}{\textsc{prevote}}
\newcommand{\precommit}{\textsc{precommit}}

\newcommand{\refround}{refRound_p}
\newcommand{\txs}{refValue_p}

\definecolor{mylightgray}{gray}{0.8}
\definecolor{mydarkgray}{gray}{0.6}
\definecolor{offline}{gray}{0.7}
\newcommand{\algemph}[1]{\colorbox{mylightgray}{#1}}

\newtheorem{dilemma}{Dilemma}
\newtheorem{requirement}{Requirement}

\newtheorem{definition}{Definition}
\newtheorem{lemma}{Lemma}
\newtheorem{theorem}{Theorem}

\usepackage{tikz}

\makeatletter
\def\@copyrightspace{\relax}
\makeatother

\algrenewcommand\algorithmicindent{2.0em}

\newcommand{\myindent}{\hspace{\algorithmicindent}}

\usepackage{lipsum}

\begin{document}

\title{Back to the Future: Rethinking Endorsement in Order-Execute Blockchains} 

\author[1,2]{Rongji Huang}
\author[1]{Yifeng Ye}
\author[3]{Gerui Wang}
\author[3]{Mingchao Wan}
\author[3]{Yuxing Duan}
\author[4]{Jingjing Zhang}
\author[1]{Guangtao Xue}
\author[1]{Shengyun Liu}
\affil[1]{\it Shanghai Jiao Tong University}
\affil[2]{\it Shanghai Academy of Future Internet Technology} 
\affil[3]{\it Beijing Academy of Blockchain and Edge Computing}
\affil[4]{\it Fudan University}

\maketitle

\begin{abstract}
Due to regulatory compliance and governance management, modern (permissioned) blockchains require flexible endorsement, which allows the endorsement policy for each contract or state object to be individually defined.
To enable flexible endorsement, Hyperledger Fabric employs an \emph{execute-order-validate} (EOV) paradigm, in which transactions first undergo speculative execution and endorsement, and are only then ordered and validated.
Meanwhile, most blockchain systems, including the platform targeted in this work (i.e., ChainMaker), still follow a conflict-free \emph{order-execute} framework.
We argue that the EOV paradigm still faces several limitations, notably high abort rates in high-contention workloads such as those in Decentralized Finance (DeFi).

To avoid refactoring our system and better suit DeFi applications,
we try to integrate flexible endorsement into the classical order-execute architecture and accordingly propose a new framework.
The key challenge is to deterministically remove problematic transactions from an ordered list, while preserving censorship resistance and decentralization for the remaining ones.
We instantiate this framework on top of Tendermint, a seminal Byzantine fault-tolerant (BFT) protocol adopted in our system, and thereby propose \sys.
By elegantly embedding endorsements into consensus, \sys incurs no additional messaging overhead in the normal case.
Empirical evaluation using an Ethereum USDT workload demonstrates that \sys achieves up to $10.6\times$ speedup in throughput over an EOV simulation on the same platform.

\end{abstract}

\vspace{-5pt}
\section{Introduction}
\vspace{-5pt}

Decentralization is fundamental to the prosperity of blockchain ecosystems~\cite{uniswap,backed,usdt}.
However, this advantage comes with an inevitable cost: the technology can be easily exploited for illicit activities~\cite{criminal,laundering,illicitreport}.
Even without malicious intent, accidents~\cite{paxosaccident,tetheraccident,blockfiaccident} can also erode the trust that underpins blockchain systems.

With the recent establishment of regulatory frameworks~(e.g., the U.S. GENIUS Act~\cite{genius} and the Hong Kong Stablecoins Ordinance~\cite{ordinance}), it has become more evident that regulation is essential for the widespread adoption of decentralized systems.
In this work, we aim to integrate proactive regulation into ChainMaker~\cite{chainmaker}, an enterprise-grade blockchain platform developed in part by some of the authors and rolled out across dozens of state-owned enterprises in China.
Toward this goal, we found that a key step is to support heterogeneous roles in the \emph{validation} of specific transactions and their execution results.
A similar idea, also referred to as \emph{flexible endorsement}, was first introduced in Hyperledger Fabric~\cite{fabric}~(hereafter, Fabric).

To support customizable endorsement policies and overcome limitations inherent in the traditional state machine replication~(SMR) paradigm~\cite{smr}, Fabric adopts a novel \emph{execute-order-validate} (EOV) framework inspired by Eve~\cite{eve}.
In Fabric, transactions are first simulated (i.e., pre-executed) in parallel by endorsers\footnote{In Fabric, endorsers are nodes that maintain local copies of the ledger and are responsible for executing and endorsing transactions.}, producing read/write sets.
Compliant transactions and their execution results are then endorsed (through signatures).
Based on predefined endorsement policies, an ordering service sequences properly endorsed transactions, typically via a consensus protocol.
As transactions were pre-executed in parallel, a final validation phase checks for read-write conflicts and aborts conflicting transactions~\cite{linearabort}.
To be committed, aborted transactions must be retried.

Despite Fabric's appealing features and proven enterprise success, the classical order-execute framework continues to dominate blockchain systems~\cite{cosmos,ResilientDB,mytumbler,bitcoin,ethereum,algorand}.
While workarounds exist 
for most limitations addressed by Fabric, supporting flexible endorsement presents inherent difficulties.
Most notably, as endorsement policies may implement arbitrary logic, whether a transaction can be properly endorsed goes beyond a consensus problem.
\emph{A key benefit} of the EOV paradigm is its simplicity in handling endorsement,
as any transaction remains off the critical path of others until getting endorsed\footnote{This is likely the main reason why Fabric adopted an ``inverted'' structure.}.
Nonetheless, placing endorsement ahead still entails several limitations, rendering it unsuitable for modern order-execute blockchains.
\begin{itemize}[topsep=-0.1em,itemsep=-0.1em,itemsep=-0.5em]
\item 
Fabric inverts the order-execute model and explicitly relies on a validation phase to ensure serializability;
Its architectural idea does not transfer readily to order-execute blockchains, as it affects not only consensus but also execution semantics;
\item 
As there is no sequencing in the execution phase, 
each transaction is endorsed individually, which hinders batching---a key technique for amortizing the cost of cryptographic operations~\cite{Basil} for endorsement; and,
\item 
The EOV framework is ill-suited for high-contention workloads, which are prevalent in financial applications like decentralized exchanges (DEX)~\cite{uniswap};
The problem not only degrades performance but may undermine liveness when certain transactions accessing hot-spot state are repeatedly aborted~(see Figure~\ref{fig:conflict}).
\end{itemize}

Numerous improvements~\cite{xox,understanding,ConChain,fabricCRDT,frabirsharp,fastFabric,fabric++,htfabric,sparsepeer} have been made to Fabric, yet they remain constrained by the original framework and inherited limitations.
We instead pivot toward the alternative paradigm.

We observe that the key to enabling certain participants to play critical roles in validating specific transactions is \emph{to keep the rest of the system decentralized and as unaffected as possible}\footnote{In principle, endorsement policies should themselves be (partially) decentralized, although this is ultimately the responsibility of application designers.};
Otherwise, one might simply adopt a centralized approach.
Targeting predominant order-execute blockchains,
we study a framework termed \emph{order-execute-endorse}, in which endorsement occurs only after transactions have been ordered and executed.
The rationale behind our design is simple: 
Rather than simplifying endorsement handling at the cost of refactoring and introducing conflicts, our framework aligns with most existing systems and remains conflict-free.

\noindent\textbf{Challenges.}
We however must deal with the cases where some transactions fail to obtain endorsements.
Since such problematic transactions may affect the outcomes of subsequent ones, 
nodes must \emph{remove} or eventually commit them in a deterministic manner.
More importantly, subsequent transactions should be re-executed and re-endorsed.
The cycle of removal and re-execution may repeat iteratively until all remaining transactions are endorsed.
As transactions are removable, we should also address the issue of endorsable transactions being deliberately removed by malicious nodes~(i.e., censorship attacks).
It is worth noting that early versions~\cite{structuredive} of Fabric (i.e., v0.6~\cite{fabric0.6} and earlier) were built on the order-execute paradigm, yet they lacked support for flexible endorsement.
This architectural shift highlights the difficulty of introducing flexible endorsement into an order-execute blockchain.

We initially attempted to introduce an endorsement phase as a standalone component, similar to Fabric's modular architecture, but found the resulting design too complicated for practical implementation and real-world deployment.
As transaction removal and re-execution may be performed iteratively, such a modular design gives rise to nested rounds of processing: one for consensus (e.g., view changes) and the other for endorsement.
We observe that the need to remove transactions and reprocess the remaining ones bears similarities with the need to advance consensus in a round-by-round fashion.
We thus try to elegantly integrate flexible endorsement into consensus messages.
To this end, 
\begin{itemize}[topsep=-0.1em,itemsep=-0.1em,itemsep=-0.5em]
\item The round (or view) advancement not merely facilitates leader election, but also aligns with the re-execution and re-endorsement cycle of flexible endorsement;
\item The (multi-phase) message exchange in each round not merely reaches agreement on transactions and their results, but also propagates endorsements and achieves consensus on the status of each transaction (whether endorsed or removed); and,
\item Timeout events not merely trigger leader elections, but also prompt the removal of unendorsed transactions.
\end{itemize}

We incorporate these ideas into Tendermint~\cite{tendermint}, whose message pattern derives from PBFT~\cite{pbft} and serves as the foundation for numerous BFT protocols~\cite{hq,upright,bftsmart,700bft,xft,sbft,hotstuff,mir,neobft,chitu}. 
Tendermint is adopted by many blockchain systems, including BNB Beacon Chain~\cite{bnb-tendermint}, Cosmos~\cite{cosmos} and ChainMaker.
We name our new protocol \sys.
To the best of our knowledge, \sys is the first BFT protocol that supports flexible endorsement while remaining conflict-free.
\sys also provides several key advantages.
\begin{itemize}[topsep=-0.1em,itemsep=-0.1em,itemsep=-0.5em]
\item The message exchange in each round does not exceed Tendermint's normal procedure.
Moreover, under a correct leader, with all transactions endorsable and synchronous network, \sys can commit a proposal (also) in three message delays.
\item \sys provides a fast path that can rapidly remove an unendorsed transaction if there is a proof that the transaction has been vetoed by its endorsers.
\item The correctness of \sys relies largely on Tendermint, which has been thoroughly analyzed~\cite{tendermintproof,uctendermint} and widely adopted.
\end{itemize}
For the above two cases, \sys also guarantees responsiveness~\cite{hotstuff}, meaning that if a sufficient number of endorsers can progress at the pace of the actual network, the entire protocol can do so as well.

We implement and evaluate \sys on top of ChainMaker, demonstrating its practical viability under production constraints, and conduct experiments on Amazon EC2.
We also simulate the EOV paradigm (\eov) based on the same codebase.
We use transfers from the Tether USD (USDT) contract~\cite{usdt} on Ethereum as our workload.
USDT is the largest stablecoin by market capitalization and plays a major role in global crypto markets.
The experiments demonstrate that compared to \eov, \sys achieves up to $10.6\times$ speedup in throughput.
We provide the correctness proof and pseudocode in Appendix~\ref{sec:correctness} and~\ref{sec:pseudocode}, respectively, where we also highlight our modifications to Tendermint.

\vspace{-5pt}
\section{Background}
\vspace{-5pt}

\subsection{System model}

We focus on the (permissioned) blockchain~\cite{fabric} or state machine replication (SMR)~\cite{smr} problem,
in which a group of $n$ nodes maintains a logically centralized ledger/database by replicating transactions and the world state across all nodes.
At most $f$ nodes can be \emph{Byzantine} faulty~\cite{byz} and exhibit arbitrary behavior, while the remaining $n-f$ nodes are assumed to be \emph{correct}.
Any blockchain solution must ensure both safety and liveness, meaning that correct nodes produce the same results when executing the same set of transactions, and that every transaction will be committed/executed eventually.

\noindent\textbf{Flexible endorsement.}
An endorsement policy specifies which nodes and how many of them must participate in the approval of a transaction and its execution results. 
Typically, an endorsement policy is associated with a contract or state object, meaning that any transaction invoking the contract or modifying the object must comply with the policy.
\begin{definition}
A transaction is considered \textbf{properly endorsed} when it has been endorsed by enough nodes specified in the corresponding endorsement policies;
A value or proposal is considered properly endorsed if all transactions within it are properly endorsed.
\end{definition}
Theoretically, in the absence of a specific policy, any set of $f+1$ nodes\footnote{The typical quorum of $n-f$ or $2f+1$ nodes is used to prevent equivocation in the ordering phase, but not for transaction validation~\cite{separating}.} 
can properly endorse a transaction, ensuring that at least one of them is correct.
We still use Fabric's monotone logical expression on sets to define each endorsement policy (see Figure~\ref{fig:structure} for an example).
Unlike Fabric, we also use this expression to verify whether a transaction is opposed by enough endorsers, in order to \emph{rapidly remove} a pending transaction.
When endorsers are not malicious, a transaction cannot be both properly endorsed and eligible for rapid removal.
Even if malicious endorsers express conflicting views, they should not compromise safety or liveness.
The protocol must additionally ensure the following theorems.
\begin{theorem}{(Safety with endorsement)}
\label{the:endorsafe}
Every committed transaction must be properly endorsed.
\end{theorem}
\begin{theorem}{(Liveness with endorsement)}
\label{the:endorlive}
Every properly endorsed transaction that cannot be rapidly removed will be committed eventually.
\end{theorem}

\noindent\textbf{Timing assumptions.}
We assume the underlying network is eventually synchronous; that is, after some unknown time, called the global stabilization time (GST), message delay between correct nodes is bounded by $\Delta$~\cite{partialsynchrony}.
The total number $n$ of nodes is thus at least $3f+1$.
In this work, synchrony assumptions serve two purposes: (1) to ensure the liveness inherited from vanilla Tendermint and (2) to provide censorship resistance through the timely propagation of endorsements and conflicting messages.
The safety property is independent of synchrony assumptions.

\vspace{-5pt}
\subsection{Existing frameworks}

Blockchain systems, either permissioned~\cite{corda,besu,mytumbler,fiscobcos} or permissionless~\cite{bitcoin,ethereum,algorand,sui}, widely adopt a classical order-execute architecture (see Figure~\ref{fig:frame0}).
This paradigm is typically embedded within a full-fledged state machine replication (SMR) protocol~\cite{smr}, in which a proposer or primary establishes an order that other nodes validate and agree upon, commonly along with execution results and ledger updates.
Although the order-execute paradigm retains its simplicity, it is arguably not suitable for (permissioned) blockchain contexts due to limitations such as non-determinism, deficiency of sequential execution, and lack of flexible trust model~\cite{fabric}.

Fabric~\cite{fabric} 
employs a new execute-order-validate (EOV) architecture to mitigate the above-mentioned limitations (see Figure~\ref{fig:frame1}).
Fabric requires endorsers to first speculatively execute them in parallel and endorse the results.
By placing execution and endorsement ahead, a subsequent ordering component (e.g., Raft~\cite{raft} or BFT-SMaRt~\cite{bftsmart}) can reach agreement also on execution results and state updates.
Such a design decision not only enables parallel processing and resolves non-determinism, but further makes possible
customizing endorsement policies, which is of the essence for modern (permissioned) blockchains.

\begin{figure}
    \centering
    \subfloat[The traditional state machine replication framework.]{\includegraphics[scale=0.45]{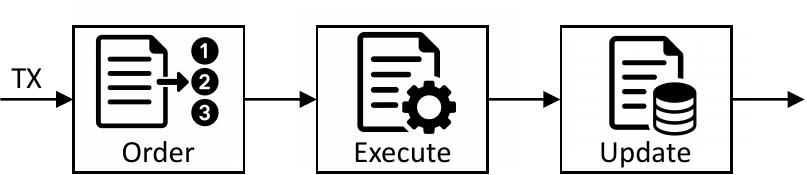}\label{fig:frame0}}\\
    \subfloat[The execute-order-validate framework (Fabric).]{\includegraphics[scale=0.45]{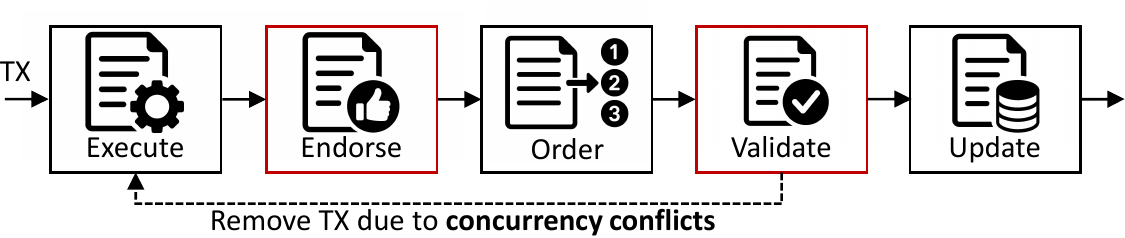}\label{fig:frame1}}
    \caption{Existing frameworks.}
    \label{fig:framework}
    \vspace{-15pt}
\end{figure}

Those benefits come with (possibly large) costs.
Most notably, as conflicting transactions may be executed in parallel,
Fabric further resorts to a validation phase to abort those that may break linearizability.
In unfavorable cases, a large fraction of conflicting transactions may be repeatedly aborted~\cite{vegeta}.
Figure~\ref{fig:txaccess} illustrates the access pattern of the USDT workload derived from Ethereum and used in our evaluation.
The workload exhibits high skewness; for example, the top $0.01\%$ of addresses are accessed by $39.23\%$ of all transactions.

After years of rapid advancement, it has become clear that many of the limitations inherent in order-execute blockchains can be effectively addressed or mitigated, enabling this paradigm to remain predominant.
For instance, to deal with non-determinism, many systems (e.g., Ethereum~\cite{ethereum} and ChainMaker) also require the consensus layer to reach agreement on execution results and updated state, thereby ensuring that no divergent state can arise.
To enable parallel processing, proposers can naturally exploit multi-thread~\cite{rex,deterschedule,Anjana,rolis} or even multi-worker~\cite{tusk} execution and have other nodes schedule execution in the same way.
Among the limitations addressed by Fabric, a notable exception is flexible endorsement, which can substantially complicate the transaction commit process.
We will elaborate on its challenges in \S\ref{sec:paradigm}.

\begin{figure}[tp]
    \centering
    \includegraphics[width=0.7\linewidth]{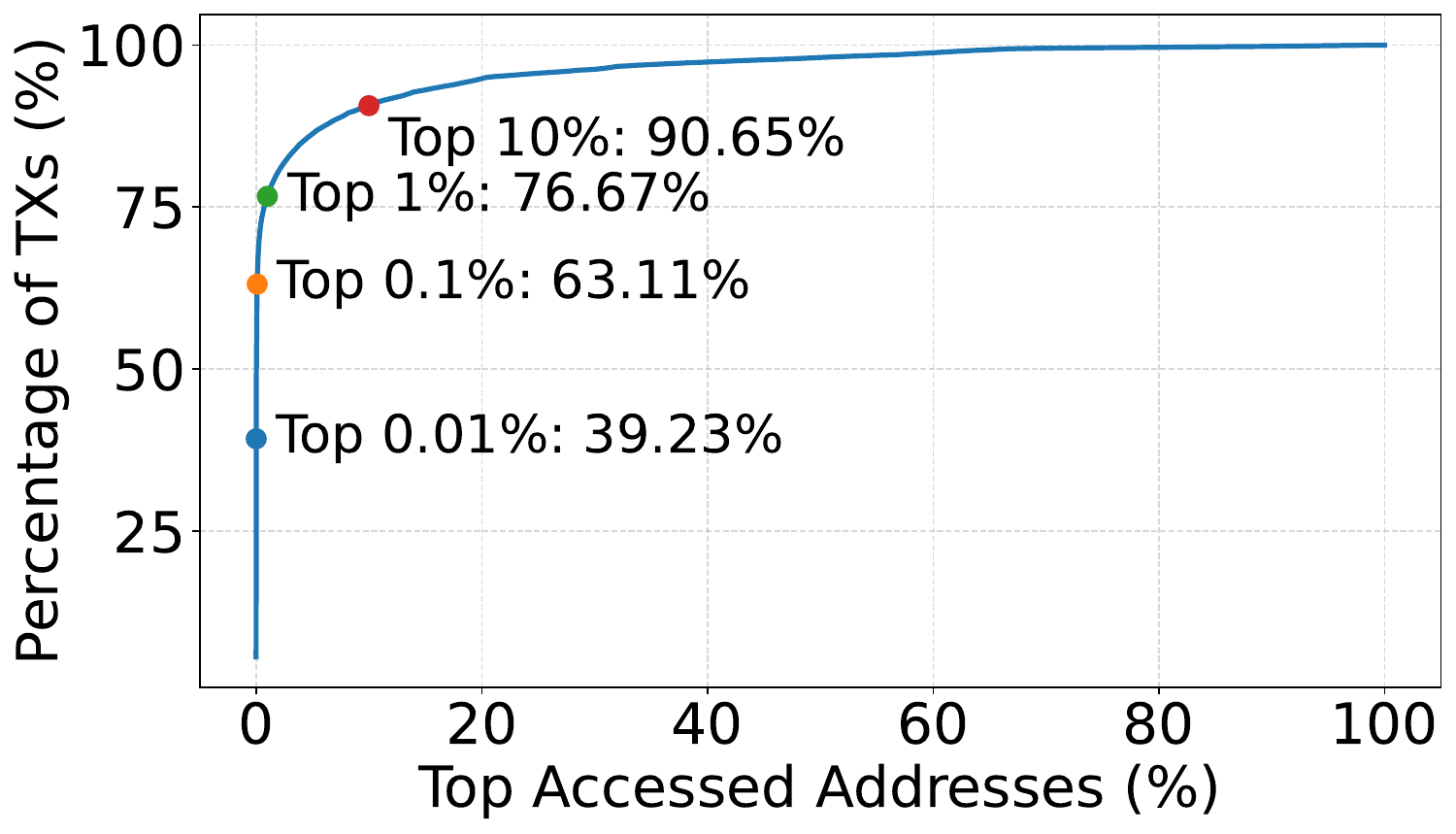}
    \caption{Cumulative distribution of transactions accessing USDT addresses on Ethereum between block heights 23,700,767 and 23,722,235.}
    \label{fig:txaccess}
    \vspace{-15pt}
\end{figure}

\vspace{-5pt}
\subsection{PBFT and Tendermint}
\label{sec:tendermint}

PBFT~\cite{pbft} is the first practical BFT SMR protocol in partially synchronous networks. 
With a single primary, PBFT offers an efficient three-phase commit in normal-case operations. 
Its complexity is dominated by the view-change procedure~\cite{hotstuff}, which replaces a potentially faulty primary.
Inspired by PBFT, Tendermint~\cite{tendermint} is a round-based protocol that adapts PBFT's three-phase commit process to decentralized and open environments.
Unlike PBFT, Tendermint employs a locking mechanism to safeguard the potentially committed proposal, and augments each phase with timers to guarantee advancement to the next round.
In this way, Tendermint obviates the need for complicated view-change operations, ensuring smooth progress and simplifying its implementation. 
Tendermint is also greatly simplified by its assumption of a \emph{gossip-based} communication model, which ensures that any message delivered to a correct node will eventually be delivered to all correct nodes.

The message pattern of Tendermint is shown in Figure~\ref{fig:tendermint}.
At the beginning of round $r$ (i.e., the \propose{} phase), the primary of round $r$ proposes a value $v$ and propagates it to others. 
Upon receiving $v$, each node $p$ verifies whether $v$ matches its locked value or $v$ represents a more recent quorum-certified (QC) value\footnote{
A QC value is one that has been approved (via \prevote{} messages) by a quorum of $\lfloor\frac{n+f}{2}\rfloor +1$ nodes (e.g., $2f+1$ when $n=3f+1$), while a more recent QC value is one generated in a higher round.
}. 
If so, it broadcasts a \prevote{} message to indicate its approval and proceeds to the \prevote{} phase.
Once two-thirds of nodes have prevoted for $v$ (i.e., $v$ is quorum-certified), node $p$ in the \prevote{} phase updates its locked value to $v$, broadcasts a \precommit{} message and enters the \precommit{} phase.
Finally, upon receiving \precommit{} messages from two-thirds of nodes,
$v$ becomes committed. 

To provide liveness, whenever a timeout occurs, node $p$ either broadcasts a $nil$ message (\prevote{} $nil$ or \precommit{} $nil$, when $tmr_{P,r}$ or $tmr_{V,r}$ expires), which helps nodes to move to the next phase, or directly transitions to the next round (when $tmr_{C,r}$ expires).
Upon entering a new round $r+1$, the new primary of round $r+1$ should propose the most recent QC value it is aware of, updated whenever a more recent QC value is received.

\vspace{-5pt}
\section{The order-execute-endorse framework}
\vspace{-5pt}
\label{sec:paradigm}

Although most limitations inherent in the order-execute framework can be mitigated, the integration of flexible endorsement remains non-trivial.
On the one hand, flexible endorsement must, in principle, support all types of policies that may implement arbitrary logic (see an example in Figure~\ref{fig:structure}).
This policy is in sharp contrast to quorum requirements in consensus, where any one-half or two-thirds of nodes are sufficient.
With respect to safety and liveness properties, endorsement policies may impose stricter requirements than those introduced by consensus (see an example in Figure~\ref{fig:challenge}).

On the other hand, as the ordering phase establishes a sequence for blocks and transactions, any \emph{pending} transaction that is waiting for endorsements (e.g., $tx_0$ in Figure~\ref{fig:structure}) suspends subsequent ones (e.g., $tx_3$ and $tx_6$).
Whether or not a pending transaction can eventually get properly endorsed is itself a consensus problem or even beyond.
Besides the challenges mentioned above,
flexible endorsement must satisfy the following requirement.

\begin{figure}[tp]
    \centering
    \includegraphics[width=0.8\linewidth]{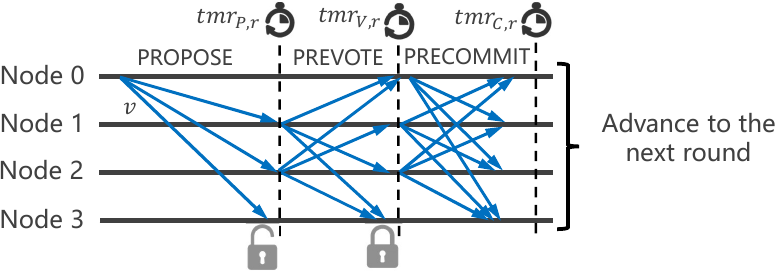}
    \caption{The message pattern of Tendermint.}
    \label{fig:tendermint}
    \vspace{-15pt}
\end{figure}

\begin{requirement}(Endorsement)
\label{chl:endorsement}
If a transaction is committed, its \textbf{execution results} must be properly endorsed with respect to some specific policies. 
\end{requirement}
Requirement~\ref{chl:endorsement} emphasizes that endorsements should apply to execution results rather than merely to the transactions themselves, 
a principle essential to enabling real-time governance, regulation, and risk management.
For instance, an auditor must be involved in an interest rate swap contract~\cite{fab-endorsement,hedgebook} to ensure that the swap is properly priced at fair value, which fluctuates over time as market conditions and other factors change.
Fabric also adheres to Requirement~\ref{chl:endorsement}, with peers executing transactions first and subsequently endorsing results.

With the above-mentioned considerations, it seems more pragmatic to first obtain endorsements before adding transactions to the ordered list---the approach adopted by Fabric\footnote{As the authors state, ``an endorser can simply abort an execution according to a local policy ... such unilateral abortion of execution is not possible in order-execute architectures''~\cite{fabric}.}.
We accordingly summarize the following dilemma:
\begin{dilemma}
\label{dilemma}
For any two transactions $tx_1$ and $tx_2$: (1) if they are ordered (say, with $tx_1$ preceding $tx_2$) prior to the endorsement of their execution results, then $tx_2$ may be delayed due to the unresolved status of $tx_1$; or (2) if their execution results are endorsed prior to ordering, then one of them may be removed because of concurrency conflicts.
\end{dilemma}

\begin{figure}[tp]
    \centering
    \includegraphics[width=1.00\linewidth]{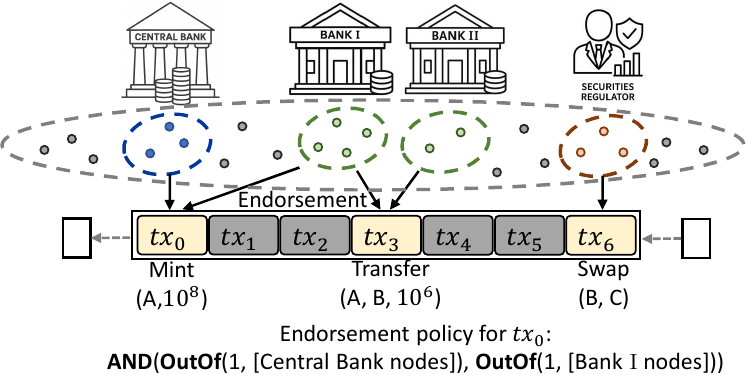}
    \caption{
    An illustrative example of flexible endorsement inspired by real-world applications.
    Within a block, a coin-minting transaction $tx_0$ must be endorsed by one node representing the central bank and another representing the commercial bank;
    a cross-border transfer $tx_3$ must be endorsed by the corresponding banks; and, a swap transaction $tx_6$ must be endorsed by the securities regulator.
    If $tx_0$ is removed due to missing endorsement, the execution results of $tx_3$ and $tx_6$ may change accordingly.
    Other transactions are not subject to any specific endorsement policy and therefore should never be removed.
    Moreover, the endorsers of any transaction should not affect the outcome or status of unrelated transactions.
    }
    \label{fig:structure}
    \vspace{-15pt}
\end{figure}

Dilemma~\ref{dilemma} indicates that either we choose to avoid concurrency conflicts but introduce delays between (conflicting) transactions, or we choose to avoid delays caused by pending endorsements but may suffer from concurrency conflicts.
While Fabric adopted the second way, 
this work instead commits to the first path that can seamlessly plug into existing order-execute blockchains,
which we refer to as the \emph{order-execute-endorse} framework (see Figure~\ref{fig:frame2}).
Specifically, transactions go through the following phases:
\begin{enumerate}[topsep=-0.1em,itemsep=-0.1em,itemsep=-0.5em]
\item[\textbf{1.}] Transactions are batched and sequenced by an ordering component, which may optionally provide fairness guarantees~\cite{pompe,probfair};
\item[\textbf{2.}] Transactions within each block are executed following the order given by the ordering phase;
\item[\textbf{3.}] Execution results are endorsed by the corresponding endorsers; 
if any transaction is removed due to missing endorsements, its subsequent transactions~(within the same block) must be (re-)executed and (re-)endorsed;
\item[\textbf{4.}] Once the remaining transactions are properly endorsed and committed, the (world) state is updated.
\end{enumerate}

To satisfy Requirement~\ref{chl:endorsement}, we must re-execute and re-endorse the subsequent transactions in Phase 3, as removing one transaction may alter the execution results of its successors (see Figure~\ref{fig:structure}).

The order-execute-endorse framework can be instantiated with any consensus protocol.
However, to address the challenges discussed above (and shown in Figure~\ref{fig:challenge}), the following requirement must be satisfied.

\begin{requirement}(Agreement)
\label{chl:agreement}
Nodes should reach agreement on whether a transaction has collected enough endorsements.
\end{requirement}

\begin{figure}[tp]
    \centering
    \includegraphics[width=0.95\linewidth]{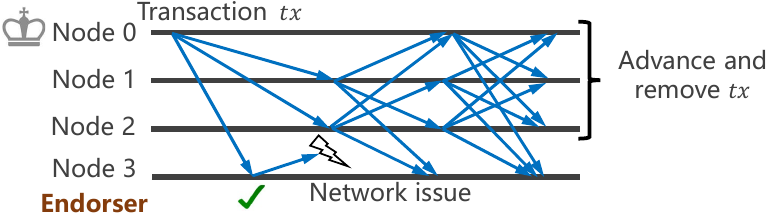}
    \caption{
	The problem of \emph{trivially} combining flexible endorsement and a consensus protocol.
    Consider the message pattern of Tendermint.
    Node 3 serves as the sole endorser for transaction $tx$.
    It has endorsed $tx$ but encountered a network issue.
    Other nodes proceed and eventually remove $tx$ (presuming that node 3 is faulty), even though, from the perspective of node 3, $tx$ would have been committed under the standard Tendermint procedure.}
    \label{fig:challenge}
    \vspace{-12pt}
\end{figure}

\begin{figure}[tp]
    \centering
    \includegraphics[scale=0.47]{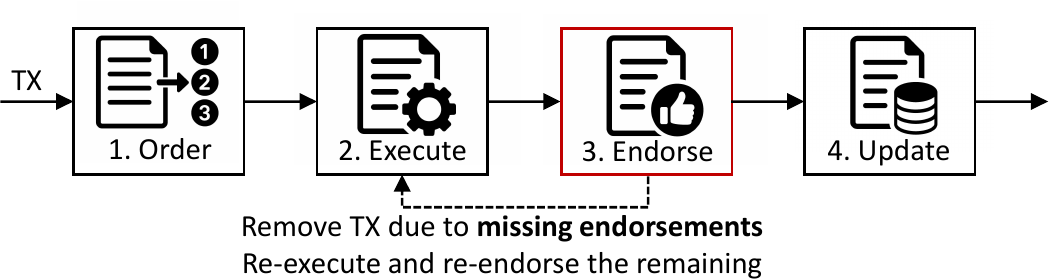}
    \caption{The order-execute-endorse framework (\sys).}
    \label{fig:frame2}
    \vspace{-10pt}
\end{figure}

\begin{figure*}[tp]
    \centering
    \includegraphics[width=1.00\linewidth]{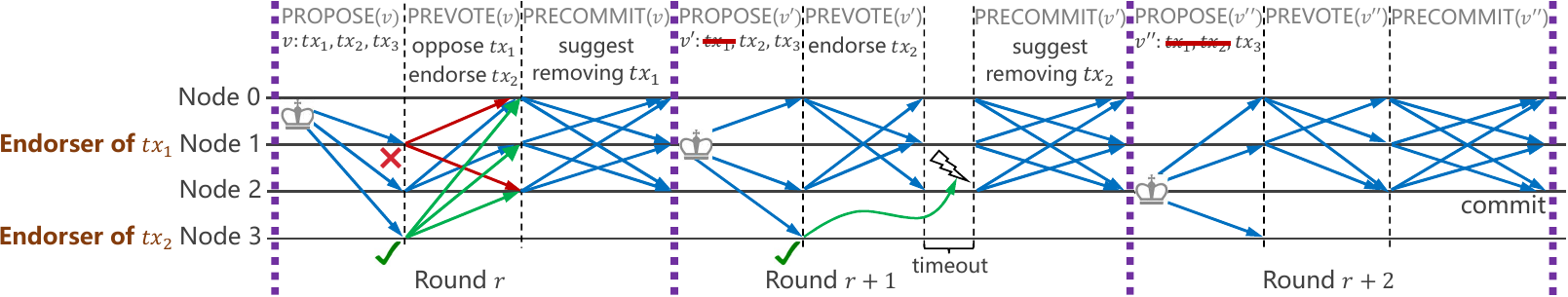}
    \caption{An example for \sys.
    The proposal in round $r$ contains three transactions: $tx_1$, $tx_2$, and $tx_3$.
    Since the endorser of $tx_1$ (node 1) opposes it, other nodes suggest its removal during the \precommit{} phase of round $r$, and the proposer of round $r+1$ (node 1) issues a proposal excluding $tx_1$.
    Due to network issues or processing delays, the endorser of $tx_2$ (node 3) fails to disseminate its endorsement in time, prompting other nodes to suggest removing $tx_2$ in round $r+1$.
    Finally, the proposal containing only $tx_3$ is committed in round $r+2$.}
    \label{fig:example}
    \vspace{-15pt}
\end{figure*}

One could delegate this responsibility to a standalone consensus protocol, treating it as an opaque component.
However, such an answer is burdensome, as any solution to the consensus problem is itself round-based or view-based (say, an ``inner loop'').
This implies that an additional round-by-round process (say, an ``outer loop'') must be introduced to reach final agreement on the committed transactions, with each round determining the outcome of at least one transaction.

As transactions can be removed due to missing endorsements, we must consider the other side of the coin: A transaction should \emph{not} be removed (by malicious nodes), \emph{unless} a sufficient number of endorsers agree to do so\footnote{
This problem is exacerbated when a fair-ordering component determines an initial proposal or when block creation is delegated to a separate role~\cite{pbs}.
}.
However, 
the designated endorsers may not respond in a timely manner.
It is necessary to handle the cases where some endorsers are faulty or offline.
We thus have the following \emph{relaxed} requirement.

\begin{requirement}(Censorship Resistance)
\label{chl:resist}
A transaction $tx$ should \textbf{not} be removed unless (1) there is a certificate from endorsers attesting that $tx$ should be removed; or, 
(2) its endorsers cannot respond in a timely manner.
\end{requirement}

The first condition can be met by allowing endorsers to explicitly oppose $tx$---a mechanism we call \emph{rapid removal}, while the second condition is triggered, e.g., when a timeout occurs and correct nodes have not received enough endorsements.
For the latter case, nodes should have a way to \emph{suggest} the removal of a pending transaction. 

\vspace{-5pt}
\section{The \sys protocol}
\label{sec:sys}
\vspace{-5pt}

\subsection{Overview}
\label{sec:overview}

Instead of implementing each phase in Figure~\ref{fig:frame2} by an individual component (like Fabric), we adopt the common blockchain or SMR paradigm and embed endorsement into Tendermint's communication phases.

\noindent\textbf{Main idea.}
Recall that Tendermint has two all-to-all communication phases in each round,  \prevote{} and \precommit{} (see~\S\ref{sec:tendermint}).
To satisfy Requirement~\ref{chl:endorsement}~(see \S\ref{sec:paradigm}), upon receiving a proposal, each node executes the included transactions, after which endorsers either endorse or oppose the corresponding transactions.
The \prevote{} message serves to capture endorsers' views on transactions.
To satisfy Requirement~\ref{chl:agreement},
once a node observes enough endorsements for all transactions within a block, it notifies others via its \precommit{} message.
On the contrary, upon timeout, if any transaction $tx$ has not collected enough endorsements, the node \emph{suggests} removing $tx$ also via its \precommit{} message.
To satisfy Requirement~\ref{chl:resist}, if at least $f+1$ nodes\footnote{
Among which, at least one node is correct.
This value may increase to $2f+1$ if needed.
} suggest removing $tx$, $tx$ is \emph{removable} and the proposers of subsequent rounds may propose a new value excluding $tx$.
We further restrict that a transaction may only be removed but not added into a proposal.
Otherwise, an infinite loop may arise, repeatedly adding and removing certain transactions.

The core idea of \sys is to distinguish, but not separate, between a procedure that merely reaches agreement on a proposal but may lack enough endorsements, and a normal procedure that can also commit the proposal. 
To this end, we introduce the following definition.
\begin{definition}
\label{def:reviewed}
We say a proposal $v$ is \textbf{\reviewed} in round $r$, if $v$ has completed all three phases in round $r$; that is, some correct node has received $n-f$ \precommit{} messages for $v$.
\end{definition}

It is evident that $v$ is also quorum-certified in round $r$.
Even if a proposal is \reviewed, it may not be committed due to missing endorsements.
This stands in sharp contrast to Tendermint.
The key for liveness is to ensure the following invariant.

\begin{restatable}{invariant}{ReviewInvariant}
\label{inv:reviewed}
If a proposal $v$ is \reviewed in round $r$, then there are only two cases:
\begin{itemize}[topsep=-0.1em,itemsep=-0.1em,itemsep=-0.5em]
\item Proposal $v$ is properly endorsed; or,
\item At least one transaction in $v$ can be removed in subsequent rounds.
\end{itemize}
\end{restatable}

With invariant~\ref{inv:reviewed}, as the protocol proceeds across rounds, there eventually exists a proposal (maybe an empty one) that is properly endorsed (see an example in Figure~\ref{fig:example}). 
After this point, \sys can be considered operating in the standard Tendermint mode, with safety and liveness ensured by Tendermint.

\vspace{-5pt}
\subsection{Potential threats and solutions}
\label{sec:threat}

Even with the above-mentioned design, 
\sys remains vulnerable to three specific threats that must be addressed.
As \sys leverages Tendermint's locking mechanism for safety, the threats pertain solely to liveness and fairness.
In the following we elaborate on each threat and its solution.
It is worth noting that they all rely on Tendermint's original messages and timers.

\noindent\textbf{(1) Avoid unnecessary re-endorsement}.
Malicious endorsers could cause the system to loop indefinitely if transactions require re-endorsement each round regardless of prior outcomes.
For instance, an adversarial endorser with sole discretion over transaction $tx$ may deliberately send its endorsement at the end of the \prevote{} phase, ensuring that only $f+1$ correct nodes observe it prior to broadcasting their \precommit{} messages.
In this scenario, $tx$ cannot be removed in subsequent rounds, as only $f$ correct nodes suggest doing so.
The new proposer of round $r+1$ considers $tx$ endorsable at the beginning, but again only $f+1$ correct nodes can obtain its endorsement 
at the end of the \prevote{} phase.

To solve this problem, \emph{we prevents nodes from proposing and endorsing a new proposal if a QC value was properly endorsed in some previous round.}
To this end, each node $p$ additionally maintains a variable $refRound_p$, which indicates the round in which node $p$ observes a properly endorsed QC value.
Each proposal also carries the $refRound_*$ maintained by the proposer.
Note that even if $v$ is properly endorsed in round $r$, it may not be committed in the same round.
By referencing round $r$, the proposer in a subsequent round can make $v$ eligible for commitment.

\noindent\textbf{(2) Aggregate partial endorsements.}
Malicious endorsers may unevenly distribute their endorsements, such that each transaction has $f+1$ correct nodes with sufficient endorsements, yet no single node observes all transactions properly endorsed.
To solve this problem, we aggregate all relevant endorsements (i.e., \prevote{} messages) generated within a round.
Moreover, if a malicious endorser issues two \prevote{} messages in the same round for $v$ but with \emph{differing} endorsements, these two messages are jointly interpreted as endorsements for \emph{all} transactions in $v$.
In case nodes or applications need to preserve endorsement proofs, they retain at most two messages from each endorser.

\begin{figure}[tp]
    \centering
    \includegraphics[width=0.9\linewidth]{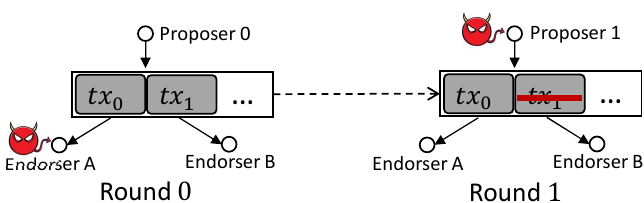}
    \caption{An example for potential censorship attacks.
    Endorser A and proposer 1 are malicious and attempt to remove $tx_1$ in round 1 by first withholding the endorsement of $tx_0$ in round 0.
    Proposer 0 and endorser B are correct.}
    \label{fig:censorship}
    \vspace{-15pt}
\end{figure}

\noindent\textbf{(3) Resist censorship.}
The third threat is somewhat more subtle.
Notably, even with a correct proposer and synchronous network, malicious or crashed endorsers can still compel the protocol to proceed to a new round\footnote{By contrast, Tendermint is able to terminate within the current round.}.
For instance (see Figure~\ref{fig:censorship}), 
assume a block contains two transactions, $tx_1$ and $tx_2$.
An adversary $\mathcal{A}$ has unilateral authority to endorse or oppose $tx_0$, while lacking such influence over $tx_1$.
In the initial round, $\mathcal{A}$ withholds its endorsement for $tx_0$ or sends it at the last moment, thereby triggering a new round.
If the new proposer is also controlled by $\mathcal{A}$, it may exploit this opportunity to remove $tx_1$ through various strategies.

A malicious proposer may simply issue a proposal excluding $tx_1$.
Although \S\ref{sec:overview} outlined an approach to mitigate this problem, batching introduces additional complexities.
As the protocol may remove multiple transactions, it is essential that each removal obtains enough supports.
To this end, \emph{$refRound_p$ is also used for 
tracking the evolvement of \reviewed{} proposals}.
That is, $refRound_p$ is updated whenever a ``shorter'' value is \reviewed at node $p$, until $p$ obtains enough endorsements for a QC value (see Figure~\ref{fig:evolvement} for an example).
At the start of each round, the proposer selects either a valid value (i.e., a properly endorsed QC value) or an \reviewed value containing the fewest transactions.

More subtly, a malicious proposer may attempt to deprive certain endorsers of their endorsement rights by withholding a proposal until the timer $tmr_{V,r}$ expires.
The proposer may still propagate its proposal to other nodes before timeout, such that the proposal can be \reviewed in this round.
In this case, even if transaction $tx_1$
can be endorsed, the absence of its endorsers still leads to its removal.
Although the gossip communication model ensures reliable delivery
within $\Delta$ time after GST, endorsers may have already sent a $nil$ \prevote{}. 

To defend against this attack, \emph{we require endorsers to always broadcast their endorsements regardless of being locked on another value or having issued a $nil$ \prevote{}.} 
Specifically, we further introduce a flag \texttt{CON} in \prevote{} messages.
A \prevote{} message with \texttt{CON} set to true means the sender also votes \emph{for} the proposal at the consensus layer, like a normal message in Tendermint.
On the contrary, if \texttt{CON} is set to false, 
the \prevote{} message only carries endorsements of the sender, but the sender votes \emph{against} the proposal at the consensus layer.
Such a message cannot contribute to the proceeding of the consensus layer.
This mechanism ensures that whenever a proposal is \reviewed,
endorsers were always able to express their opinions within the same round.

Finally, \emph{if a node receives two differing proposals in the same round, the node immediately broadcasts a $nil$ \precommit{} if it has not entered the \precommit{} phase.}
This mechanism prevents multiple proposals from being executed and endorsed within the same round, thereby mitigating DoS attacks.
We may further combine accountability~\cite{oabc} with our protocol to detect faulty processes and alleviate the problem. 
After GST, either no proposal is \reviewed in round $r$, or every correct node has received all endorsements sent by correct nodes before entering the \precommit{} phase.

\begin{figure}[tp]
    \centering
    \includegraphics[width=1.00\linewidth]{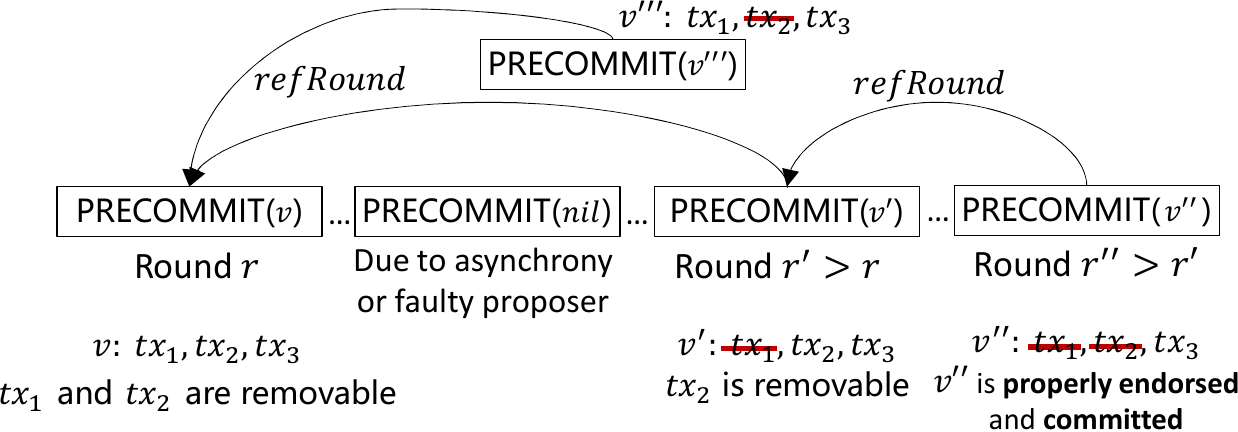}
    \caption{An illustrative example for the evolvement of \reviewed proposals.
    Proposal $v$ is \reviewed in round $r$ and can therefore be referenced by subsequent rounds.
    Proposal $v'$ references $v$ and excludes $tx_1$.
    Proposal $v''$ references $v'$, and is properly endorsed and committed in round $r''$.
    It is possible that proposal $v'''$ also references $v$ but excludes $tx_2$, in case $tx_1$ was endorsed but observed by only a few nodes before round $r'$ (before GST).
    }
    \label{fig:evolvement}
    \vspace{-15pt}
\end{figure}

With all the mechanisms described above in place, \sys guarantees the following important invariant:
\begin{restatable}{invariant}{censorship}
\label{inv:censorship}
After GST, if a sufficient number of endorsers endorse transaction $tx$ and $tx$ cannot be rapidly removed, 
then no correct node suggests removing $tx$ in round $r$.
\end{restatable}

\vspace{-5pt}
\subsection{Protocol detail}
\label{sec:detail}

Like Tendermint, each node $p$ locally maintains two critical pairs of variables to track candidate values.
The $\{lockedValue_p$, $lockedRound_p\}$ pair implements the locking mechanism, recording the QC value and its round in which node $p$ became locked.
The $\{validValue_p$, $validRound_p\}$ pair records the most recent QC value and its associated round received by node $p$, enabling $p$ to propose a valid value when acting as a new proposer.
In \sys, these two pairs are updated \emph{only} when the corresponding endorsements are also received, i.e., when the value is properly endorsed.

In addition, each node $p$ also maintains a new pair $\{refValue_p$, $refRound_p\}$, which serves two purposes (as discussed in \S\ref{sec:threat}).
If there exists a properly endorsed QC value, i.e., $validValue_p\neq nil$,
then $refRound_p$ records the round in which the endorsements for $validValue_p$ were generated\footnote{QC messages and endorsements may be generated in different rounds.}.
Otherwise, $\{refValue_p, refRound_p\}$ indicates an \reviewed value and its round.
As multiple \reviewed values may exist,
$refValue_p$ records the one with the fewest transactions node $p$ is aware of, thereby preserving monotonic decrease in transaction number.
Due to Invariant~\ref{inv:reviewed}, this pair assists proposer $q$ in proposing a ``shorter'' value. 
Besides a value $v$ (e.g., a list of transactions and their execution results), each proposal also includes two fields, $vr$ and $rr$, representing proposer $q$'s $validRound_q$ and $refRound_q$, respectively.

The following describes \sys phase by phase as node $p$ proceeds into round $r$, beginning at 0.
Initially, other round-related numbers are set to $-1$ and all values are set to $nil$. 
The pseudocode is postponed to Appendix~\ref{sec:pseudocode}.

\noindent\underline{\textbf{Phase I: \propose{}.}}
The proposer or primary $q$ of round $r$ first checks whether its $validValue_q$ is non-nil.
If so, it proposes $validValue_q$ along with the associated $validRound_q$ and $refRound_q$;
otherwise, if $refValue_q$ is non-nil, $q$ proposes $refValue_q$ along with $refRound_q$.
If neither is true, $q$ makes a new proposal.
Other nodes initialize a timer $tmr_{P,r}$.

Upon receiving a proposal (i.e., a \propose{} message containing $v$, $vr$  and $rr$), 
(1) if $vr\neq -1$, node $p$ verifies whether $v$ is a QC value in round $vr$ and whether round $rr$ contains sufficient endorsements for $v$;
(2) If $vr=-1$ but $rr\neq -1$, 
node $p$ instead verifies whether some value $v'$ was \reviewed in round $rr$, and whether, for each transaction $tx$ included in $v'$ but absent from $v$, at least $f+1$ \precommit{} messages from round $rr$ indicate its removal.
If either condition (1) or (2) holds, or if $rr=-1$ and node $p$'s local $refRound_p=-1$, node $p$ proceeds with Tendermint's locking mechanism:
It verifies whether $lockedValue_p=nil$, $v = lockedValue_p$, or $v$'s QC messages were generated in a round higher than $lockedRound_p$.
If so, 
node $p$ further distinguishes two cases:
(1) If $vr=-1$,
node $p$ executes transactions in $v$, endorses or opposes each of them, and broadcasts a \prevote{} message reflecting the results; or,
(2) if $vr\neq -1$, node $p$ simply broadcasts a \prevote{} message indicating its support for $v$. 
If $v$ fails any of the above verifications or if the timer $tmr_{P,r}$ expires, node $p$ broadcasts a $nil$ \prevote{}.

\noindent\underline{\textbf{Phase II: \prevote{}.}}
Upon receiving any $2f+1$ \prevote{} messages ($v$ or $nil$), node $p$ starts a timer $tmr_{V,r}$.
Upon receiving $2f+1$ \prevote{} messages and sufficient endorsements for $v$, node $p$ updates its $lockedValue_p$ to $v$ and $lockedRound_p$ to $r$, and broadcasts a \precommit{} message supporting $v$.
When $tmr_{V,r}$ expires, there are also two cases:
(1) if node $p$ has received $2f+1$ \prevote{} messages for $v$, it broadcasts a \precommit{} message for $v$ while suggesting the removal of the first transaction in $v$ that lacks sufficient endorsements; or,
(2) if node $p$ has not received $2f+1$ \prevote{} messages for any $v$, it broadcasts a $nil$ \precommit{}.

\noindent\underline{\textbf{Phase III: \precommit{}.}}
Upon receiving any $2f+1$ \precommit{} messages ($v$ or $nil$), node $p$ starts a timer $tmr_{C,r}$. 
Upon receiving $2f+1$ \precommit{} messages for $v$, there are also two cases:
(1) if none of them suggests transaction removal, node $p$ commits $v$, resets variables, and increments the block height; otherwise,
(2) if at least $f+1$ nodes suggest removing some transaction $tx$ in $v$, $validRound_p=-1$, and $length(v)<length(refValue_p)$, node $p$ sets its $\{refRound_p,refValue_p\}$ to $\{r,v\}$ and enters round $r+1$.
When $tmr_{C,r}$ expires, node $p$ also enters round $r+1$.

Note that whenever a value $v$ becomes quorum-certified in round $r'>validRound_p$ and $v'$s endorsements are also received, node $p$ updates $validValue_p$ to $v$ and $validRound_p$ to $r'$.
Node $p$ also updates $refRound_p$ to either $r'$ or $rr$ carried in the proposal, depending on which round rendered $v$ properly endorsed.
From the moment $validRound_*\neq -1$, 
\sys can be regarded as operating in the standard Tendermint mode.

\noindent\textbf{Rapid removal.}
As endorsers can explicitly oppose a transaction via their \prevote{} messages,
\sys also support a rapid removal mechanism. 
then node $p$ may immediately send a \precommit{} message that suggests removing $tx$, provided that $v$ is quorum-certified.
If, for each transaction $tx$ in $v$, $tx$ is either properly endorsed or vetoed by a sufficient number of endorsers, node $p$ can immediately send a \precommit{} message, provided that $v$ is quorum-certified.

\vspace{-5pt}
\subsection{Correctness argument}

We sketch that \sys provides correctness but postpone its detailed proof to Appendix~\ref{sec:correctness}.

\noindent\textbf{Safety.}
Because \sys relies on Tendermint's locking mechanism to preserve committed proposals across rounds, and because proposals are locked only when properly endorsed, it follows that \sys guarantees safety.

\noindent\textbf{Termination.}
When some proposal is properly endorsed, the gossip-based communication ensures that eventually every correct node $p$ updates its $validValue_p$ and $validRound_p$.
Since then, the termination property of \sys is guaranteed by Tendermint.
Due to Invariant~\ref{inv:reviewed}, each time a more recent \reviewed proposal is referenced by a new proposal, at least one transaction is removed.
As the protocol proceeds round by round, eventually there exists a properly endorsed proposal.
Invariant~\ref{inv:reviewed} holds because, for each transaction $tx$ in a proposal, either some correct node observes its endorsements or every correct node suggests its removal.

\noindent\textbf{Censorship resistance.}
Transaction $tx$ can be removed only if its proposal is \reviewed in some round $r$ and a correct node suggests removing $tx$.
Assume $tx$ cannot be rapidly removed.
Let $t$ denote the earliest time at which a correct node, say node $p$, sent a \prevote{} message (and propagated the proposal) in round $r$.
Due to the gossip-based communication model, after GST, every correct node received a proposal before $t+\Delta$.
If any correct node, say node $q$, suggested removing $tx$ in round $r$, its \precommit{} message must have been triggered by the timeout of $tmr_{V,r}$, which is set to $2\Delta$ time.
Hence, node $q$ sent its \precommit{} message after $t+2\Delta$.
As correct nodes always send their endorsements via \prevote{} messages, 
at time $t+2\Delta$ node $q$ must have received the endorsements for $tx$.
Thus, node $q$ should not suggest removing $tx$.

\vspace{-5pt}
\subsection{Endorsement semantics}
\label{sec:policy}

As transactions within a block are endorsed and agreed upon collectively, the endorsement semantics of \sys should be further refined.
In the \prevote{} phase, endorsers express their views on a transaction $tx$ in one of three ways: 
\begin{enumerate}[topsep=-0.1em,itemsep=-0.1em,itemsep=-0.3em]
\item Endorse $tx$ (by default);
\item Oppose $tx$ based on its execution results; or,
\item Oppose $tx$ irrespective of its execution results.
\end{enumerate}
The second choice can be interpreted as ``the endorser opposes $tx$ if all preceding transactions within the same block are committed''.
With the third choice, endorsers can veto multiple transactions via a single round, quickly converging on a properly endorsed proposal.

Likewise, during the \precommit{} phase, nodes provide their suggestions regarding transaction $tx$ in one of three ways: 
\begin{enumerate}[topsep=-0.1em,itemsep=-0.1em,itemsep=-0.3em]
\item Keep $tx$ (by default);
\item Remove $tx$ if $tx$ is the first removable one; or,
\item Remove $tx$ irrespective of preceding transactions.
\end{enumerate}
By the third way, the proposer in a later round can remove multiple such transactions.
A node will suggest the removal of $tx$ via the third manner if (1) a sufficient number of endorsers oppose $tx$ via the third manner, or (2) the timeout of a contract $C$'s endorsements triggers the removal of $tx$ and $tx$ will always invoke $C$ regardless of the execution contexts\footnote{It is possible that the removal of preceding transactions alters the invocation of $tx$.}.
Otherwise, nodes suggest removal via the second manner.

If multiple transactions can be removed by the second manner, the new proposer should remove the first one.
Other nodes, upon receiving a new proposal, check whether the removed transaction $tx$ is the first among all removable ones.
More specifically, for any transaction $tx'$ preceding $tx$, 
the nodes must have received the endorsements of $tx'$.
This verification does not affect liveness, since, through gossip, the new proposer disseminated all \precommit{} messages and endorsements it has received.
The new proposer also removes all transactions that are removable via the third manner.
To keep our design and implementation simple, we permit scenarios in which transaction $tx$ is properly endorsed in some round $r$ yet removed in a later round $r'>r$.
If this situation must be addressed, one may further link endorsement proofs across multiple rounds.

\section{Further Discussion and Limitations}

\sys can serve as the layer-1 infrastructure for either permissioned or permissionless blockchains (with PoX), as the consensus protocol for a shard or sidechain, or as a core component of a layer-2 solution.
In addition, \sys supports several new features.

\noindent\textbf{Multi-party atomic commit.}
For privacy reasons, multiple untrusted parties may wish to conceal their trading logic and submit only the resulting value-transfer transactions on-chain.
\sys allows multiple transactions to be treated as a `super-transaction'.
Either all related transactions are committed in the same block, or, if any of them is removed, the remaining transactions must be re-endorsed and can also be removed by endorsers.
In Fabric, as batching is not applied during the execution phase, each transaction is endorsed individually, resulting in a lack of atomicity.

\noindent\textbf{Dynamic endorsement.}
By embedding endorsements into consensus messages, \sys naturally supports \emph{dynamic} endorsement, where endorsements are needed only under certain conditions.
For example, a token contract may impose a daily limit on every account, such that endorsements are triggered only when a transfer exceeds that limit.
In contrast, Fabric requires every transaction to be endorsed by specific peers, regardless of the current state.
Dynamic endorsement provides a more natural means of combining decentralization with (partially) centralized regulation.

\noindent\textbf{On-chain accountability}.
When a transaction is removed, the decision is recorded on chain.
There are substantial differences between a transaction that was once proposed and removed, and one that never appeared in any block~(due to missing endorsements).
This record provides auditable evidence for assessing an endorser's responsiveness and reputation.
Besides, removed transactions may still pay gas fees, thereby mitigating DoS attacks.

\noindent\textbf{Limitations.}
The main limitation of \sys and its framework is that each transaction lies on the critical path of all others.
First, since risk detections can be time-consuming, the $\Delta$ duration for $tmr_{V,r}$ may be insufficient.
A naïve solution is to increase the waiting time, but this negatively impacts latency and prolongs the removal of timed-out transactions.
A more promising solution might be a \emph{two-phase} protocol.
The first phase pre-processes transactions off-chain without accurate execution results, while only the second phase takes place on chain.
As much effort as possible should be devoted to the first phase. 
Once a client receives sufficient acknowledgments from endorsers in the first phase, it submits the transaction on-chain.
In the second phase, endorsers simply verify the compliance of its execution results.

Second, Byzantine endorsers, when colluding with malicious clients, can insert several transactions into a block and remove only one per round, badly affecting performance.
Such malicious behaviors may be mitigated by, e.g., collecting a gas fee for removed transactions.
Moreover, since the transaction-removal mechanism in the \precommit{} phase is generic, nodes may quickly discard multiple suspicious transactions within a single round.
Note that such removal should not undermine liveness or censorship resistance.
We leave this direction for future work.

\vspace{-5pt}
\section{Evaluation}
\vspace{-5pt}

\begin{figure*}[t!]
    \centering
    \begin{minipage}{0.28\linewidth}
        \includegraphics[scale=0.25]{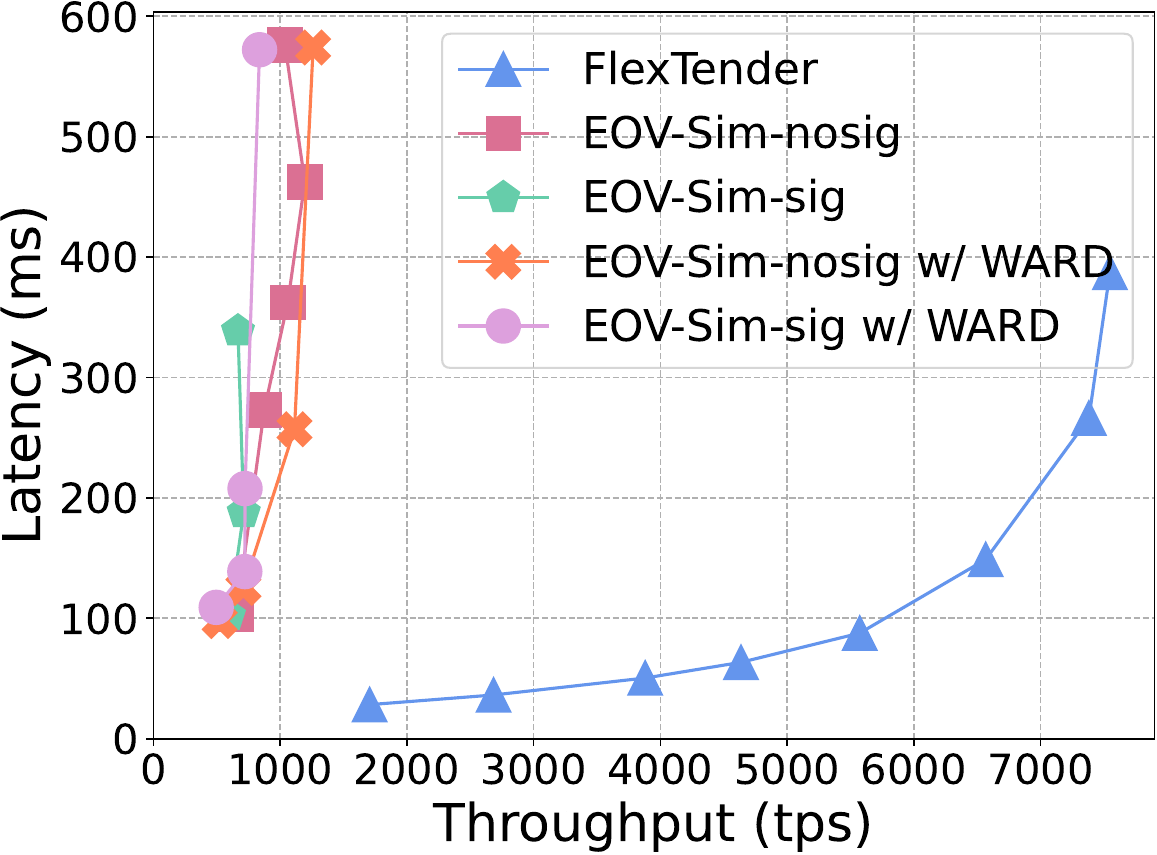}
        \caption{Performance under the USDT workload with $n=4$.}
        \label{fig:eov-flexmint-4}
    \end{minipage}
    \hfill
    \begin{minipage}{0.34\linewidth}
        \centering
        \includegraphics[scale=0.27]{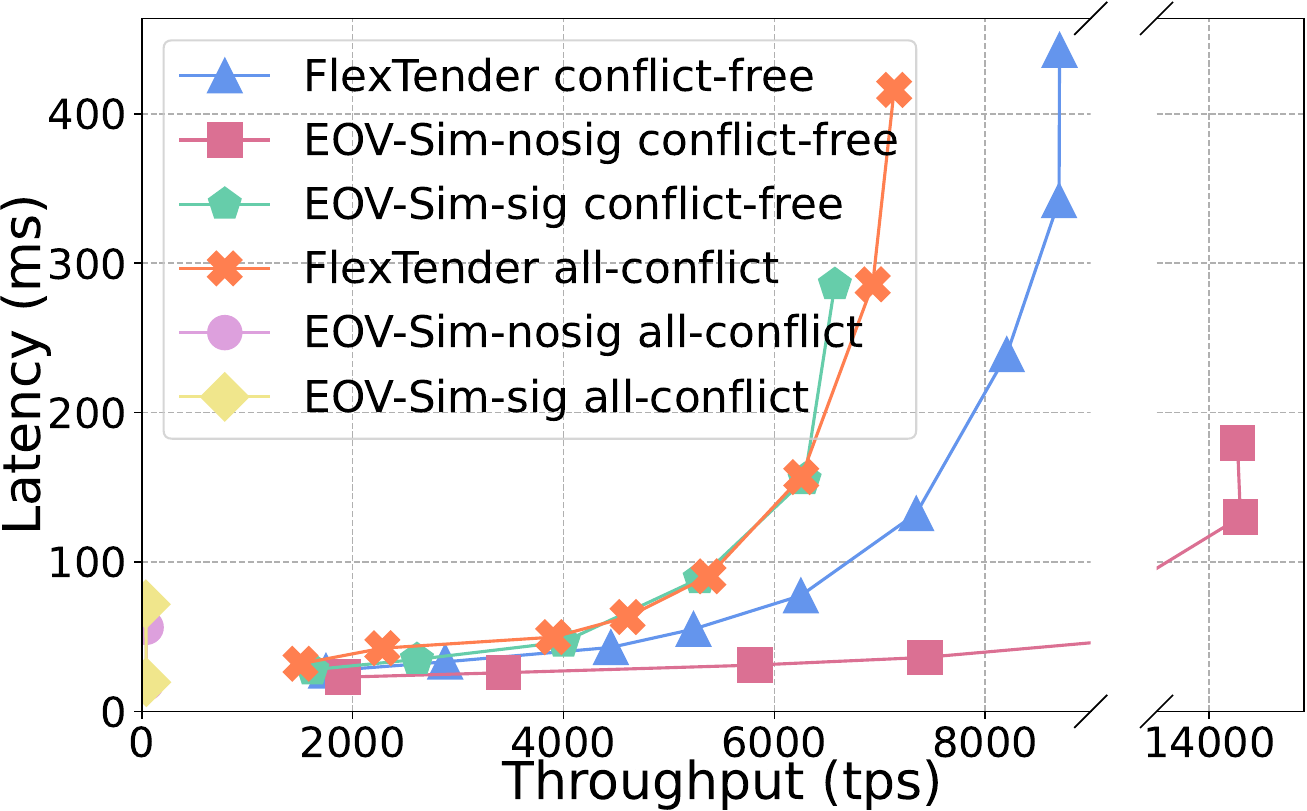}
          \caption{Performance under the conflict-free and all-conflict workloads with $n=4$.}
        \label{fig:conflict}
    \end{minipage}
    \hfill
    \begin{minipage}{0.34\linewidth}
       \centering
        \includegraphics[scale=0.26]{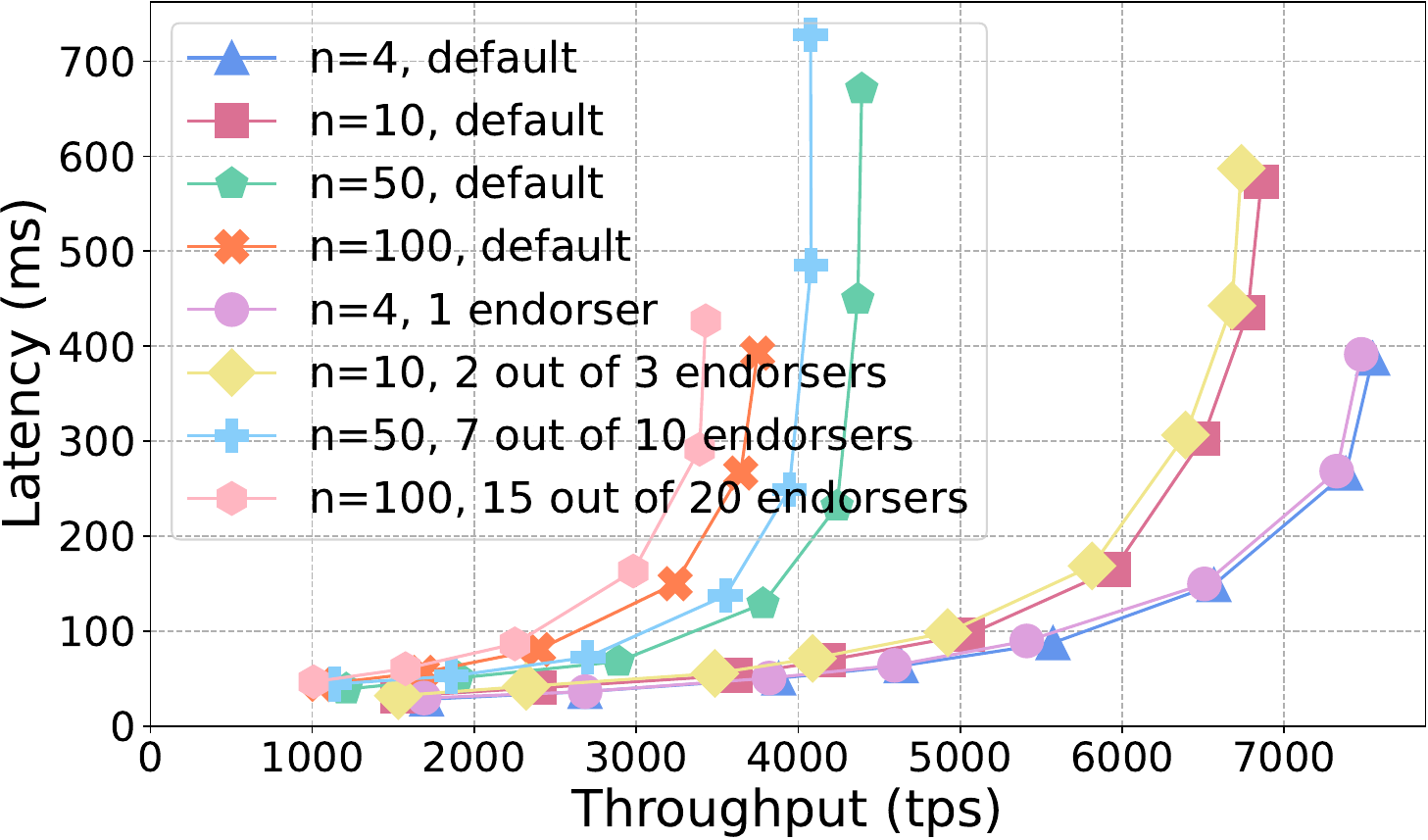}
        \caption{Performance of \sys with different node numbers and policies.}
        \label{fig:scale}
    \end{minipage}
    \vspace{-15pt}
\end{figure*}

\subsection{Implementation}

We have implemented \sys in Go within ChainMaker 2.3.6~\cite{chainmaker}, 
which is augmented with a parallel-execution engine in which the primary pre-executes transactions in parallel and produces an execution schedule.
The schedule is represented as a DAG that other nodes follow.
Messages and world state are persisted.
Note that no rollback mechanism is required, since the ledger state is not updated until all transactions within a proposal have been endorsed and committed.
The code of \sys is available at \url{https://git.chainweaver.org.cn/stcsm/flextender}.

We also simulate a simplified EOV procedure on our platform, which we refer to as \eov.
First, the primary simulates (pre-executes) transactions in parallel based on the same snapshot, producing read/write sets for each transaction.
The pre-executed transactions are then batched as a proposal and ordered by Tendermint.
After consensus, each node locally checks, for each transaction $tx$, whether $tx$ has read a key updated by a preceding transaction within the same proposal.
If so, $tx$ is aborted, and the new primary places it back into its mempool in FIFO order.
To simplify the setting while maximizing the benefit of EOV's support for independent transaction simulation by endorsers, we do not require any node other than the primary to execute transactions.
In Fabric, distinct endorsers must return consistent results before the client can proceed to the ordering phase, a requirement that is non-trivial to satisfy under high-contention workloads.
Since Fabric does not batch transactions during the execution phase, each transaction must be endorsed individually.
We further distinguish two cases based on whether each transaction is endorsed and signed by the primary (\eov-sig and \eov-nosig).
Both \sys and \eov rotate the primary as the round number or block height advances.

To reflect the state of the art in Fabric variants, we also incorporate a transaction reordering mechanism~\cite{fabric++,frabirsharp,htfabric} into \eov, whereby the primary rearranges the transaction sequence so as to avoid unnecessary aborts.
Specifically, we implemented the WARD algorithm~\cite{htfabric}, which balances computational complexity against effectiveness by reordering transactions according to their conflict degrees. 
Other Fabric variants are either largely orthogonal to our contribution (e.g., pipelining~\cite{fastFabric,sparsepeer}) or can no longer satisfy the requirement that execution results be properly endorsed (e.g., re-execution~\cite{xox,htfabric} after the validation phase).

\noindent\textbf{Optimizations.}
Because re-execution incurs additional overhead, we further introduce two simple techniques.
First, proposers avoid propagating the payload in any round $r>0$ (but only its hash), provided that the new proposal references an \reviewed one.
Second, when transaction $tx$ is removed,
the proposer only re-executes transactions that have causal dependencies on $tx$ in the DAG.
If any of them introduces new dependencies, the proposer instead re-executes all.

\vspace{-5pt}
\subsection{Experimental setup}

\noindent\textbf{Testbed.}
We conduct experiments on Amazon EC2 platform, both within a region (LAN) and across regions (WAN).
Each node is deployed on a c5.2xlarge instance equipped with 8 vCPUs, 16 GiB of memory, and up to 10 Gbps of bandwidth.

\noindent\textbf{Workload.}
We primarily use the USDT~\cite{usdt} transfers on Ethereum as our workload, as they capture key characteristics of DeFi applications and can be readily evaluated on our platform.
The workload includes both direct invocations and those triggered by other smart contracts.
For instance, a swap~\cite{uniswap} operation may contain several USDT transfers.
To preserve the actual access pattern as much as possible, we treat them as a single transaction.
The data span block heights 23,700,767 
to 23,722,235 (November 2025),
comprising a total of 1,215,353 transfers.
Each node locally maintains a mempool (transaction pool).
When acting as the primary, the node retrieves a batch of transactions from its mempool and submits them to the execution engine.
As the primary rotates every round or height, 
we evenly distribute the stream of USDT transfers across all nodes, ensuring that nodes propose transactions in the same order as they appeared on Ethereum.

\noindent\textbf{Gossip-based communication.}
Tendermint and \sys rely on gossip-based communication to provide liveness and censorship resistance.
For small-scale experiments (with $n=4$ or $n=10$), each node propagates every message to all other nodes;
for median-scale experiments (with $n=50$ or $n=100$), each node instead propagates every message to $\lceil\log_{2}n\rceil$ randomly selected nodes, following a configuration commonly used in blockchain networks~\cite{kedemliaeth} and in our platform.
Each timer is set to two seconds.

\vspace{-5pt}
\subsection{Performance in LAN}

We first conduct experiments within a region.
We gradually increase the batch size until the protocols reach saturation.
There is no specific endorsement policy in this setting, namely any two-thirds of nodes can properly endorse a transaction.
Figure~\ref{fig:eov-flexmint-4} presents the results for $n=4$.
\sys achieves $10.6\times$ speedup in throughput compared to \eov-sig.
\eov's performance is throttled by frequent aborts, as there are certain ``hot-spot'' addresses (see Figure~\ref{fig:txaccess}) on Ethereum such as those involved in swaps between USDT and other popular tokens, as well as the hot wallets of centralized exchanges\footnote{The two most frequently accessed addresses in this workload are \href{https://etherscan.io/address/0xc7bBeC68d12a0d1830360F8Ec58fA599bA1b0e9b}{0xc7bBeC68d12a0d1830360F8Ec58fA599bA1b0e9b} and \href{https://etherscan.io/address/0x559432e18b281731c054cd703d4b49872be4ed53}{0x559432e18b281731c054cd703d4b49872be4ed53}, corresponding to the Uniswap and OKX hot wallets, respectively.}.
Even with the reordering mechanism (i.e., WARD), the improvement remains marginal.
This is because all conflicts in USDT transfers stem from transactions reading from and writing to the same address, and such (read-after-write and write-after-read) conflicts can hardly be resolved through reordering.
Figure~\ref{fig:eovabort} further shows the abort rates and instantaneous throughput of \eov-nosig when saturated, where higher abort rates generally coincide with lower throughput.
Although \sys nodes also execute transactions in parallel, its retries are confined to the primary alone.
\eov-nosig and \eov-sig differ only slightly in performance.  
While signing operations introduce additional overhead, the high abort rate remains the primary bottleneck in this setting.

We then evaluate both protocols under the conflict-free and all-conflict workloads.
In the conflict-free workload, each transfer involves distinct \textsc{from} and \textsc{to} addresses, whereas in the all-conflict workload, all transfers access the same pair of addresses.
As reordering has no effect in either scenario, we omit it from the evaluation.
Figure~\ref{fig:conflict} demonstrates that \eov-nosig outperforms \sys when there is no conflict.
This is because in \eov only the primary executes transactions, analogous to how Fabric allows endorsers to simulate transactions independently.
In contrast, \sys adopts a two-phase execution model, in which the primary executes and schedules transactions first and other nodes subsequently validate the results by also executing the transactions.
The performance of \eov-sig degrades significantly under the conflict-free workload because cryptographic operations introduce non-negligible costs, highlighting the benefit of batching.
Under the all-conflict workload, both protocols exhibit reduced throughput compared to the conflict-free case, as concurrency conflicts must be handled either within the primary (\sys) or by retrying the entire procedure (\eov).
\eov can barely work because nearly all transactions are constantly aborted and retried, which exacerbates the problem posed by independent simulation.
Interestingly, \eov occupies both ends of the spectrum, whereas \sys consistently provides stable performance regardless of conflicts, reflecting their different design choices.

\begin{figure}[tp]
    \centering
    \includegraphics[scale=0.25]{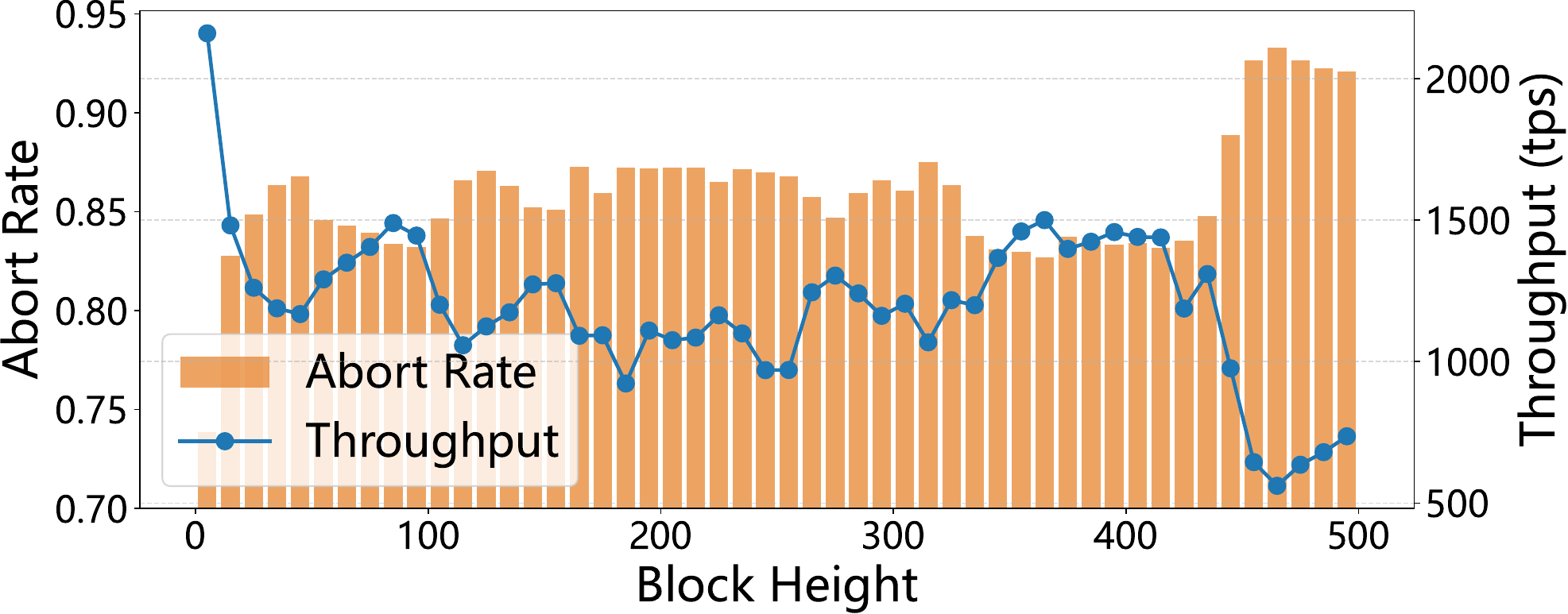}
    \caption{Abort rates and instantaneous throughput of \eov-nosig when saturated.}
    \label{fig:eovabort}
    \vspace{-15pt}
\end{figure}

We further evaluate \sys with the number of nodes ranging from 4 to 100 under different endorsement policies, as depicted in Figure~\ref{fig:scale}.
When the endorsement policy requires a single designated endorser or $t$-out-of-$m$ endorsers, \sys always incurs moderate performance overhead relative to the default policy, since nodes can proceed only after receiving endorsements from those specific endorsers rather than from any two-thirds of nodes.
As the number of nodes increases, \sys also experiences performance degradation, as the primary has to propagate proposals to more nodes and each node has to process more consensus messages~\cite{bftbottleneck}.
Large-scale performance could be further improved either by amortizing the primary's burden~\cite{Kauri,rainblock} or by incorporating sharding or layer-2 solutions, which are largely orthogonal to the main contributions of this work.

\begin{figure*}[t!]
    \centering
    \begin{minipage}{0.32\linewidth}
        \includegraphics[scale=0.20]{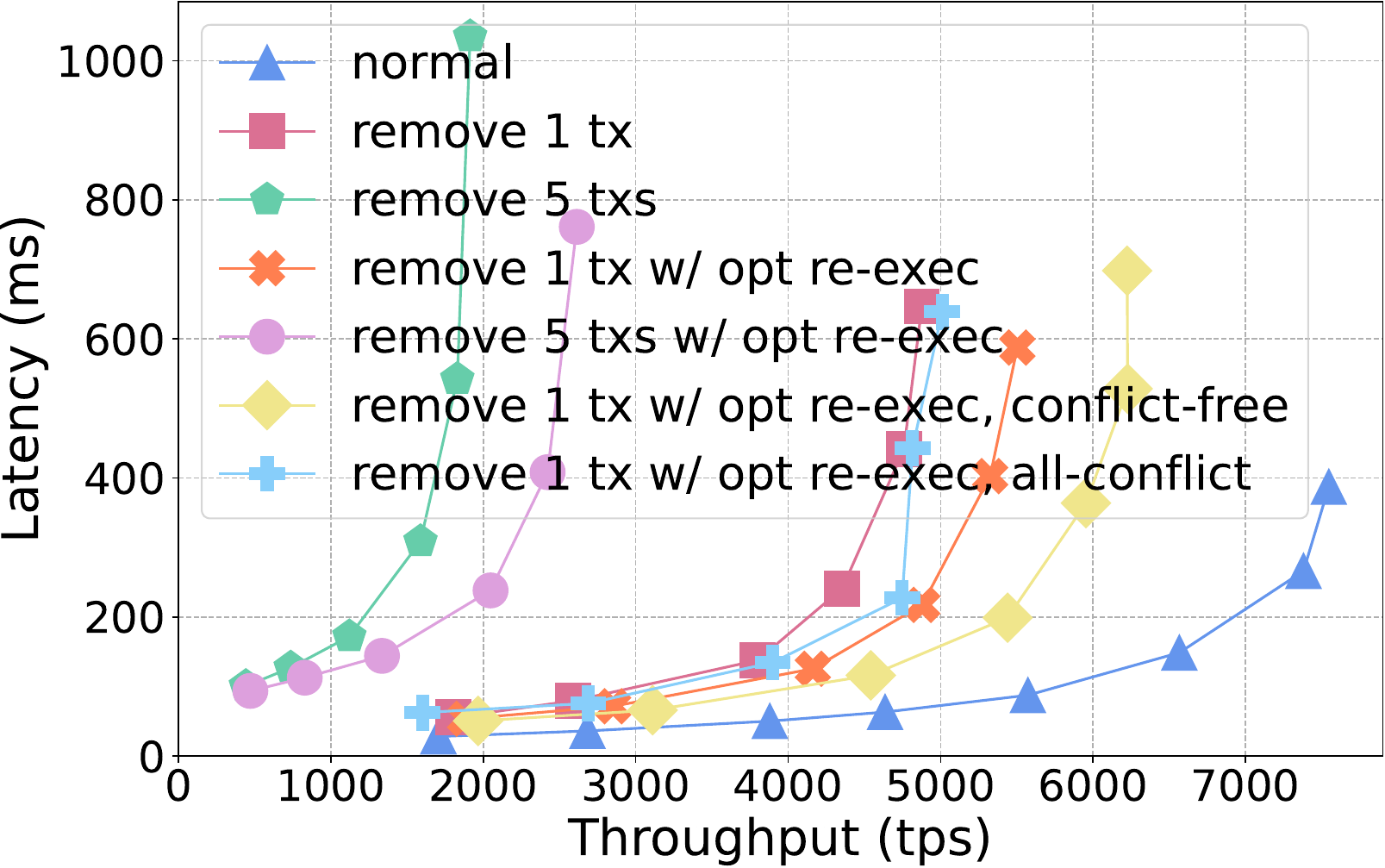}
        \caption{Performance of \sys with rapid removal ($n=4$).}
        \label{fig:fastremove}
    \end{minipage}
    \hfill
    \begin{minipage}{0.32\linewidth}
        \centering
        \includegraphics[scale=0.24]{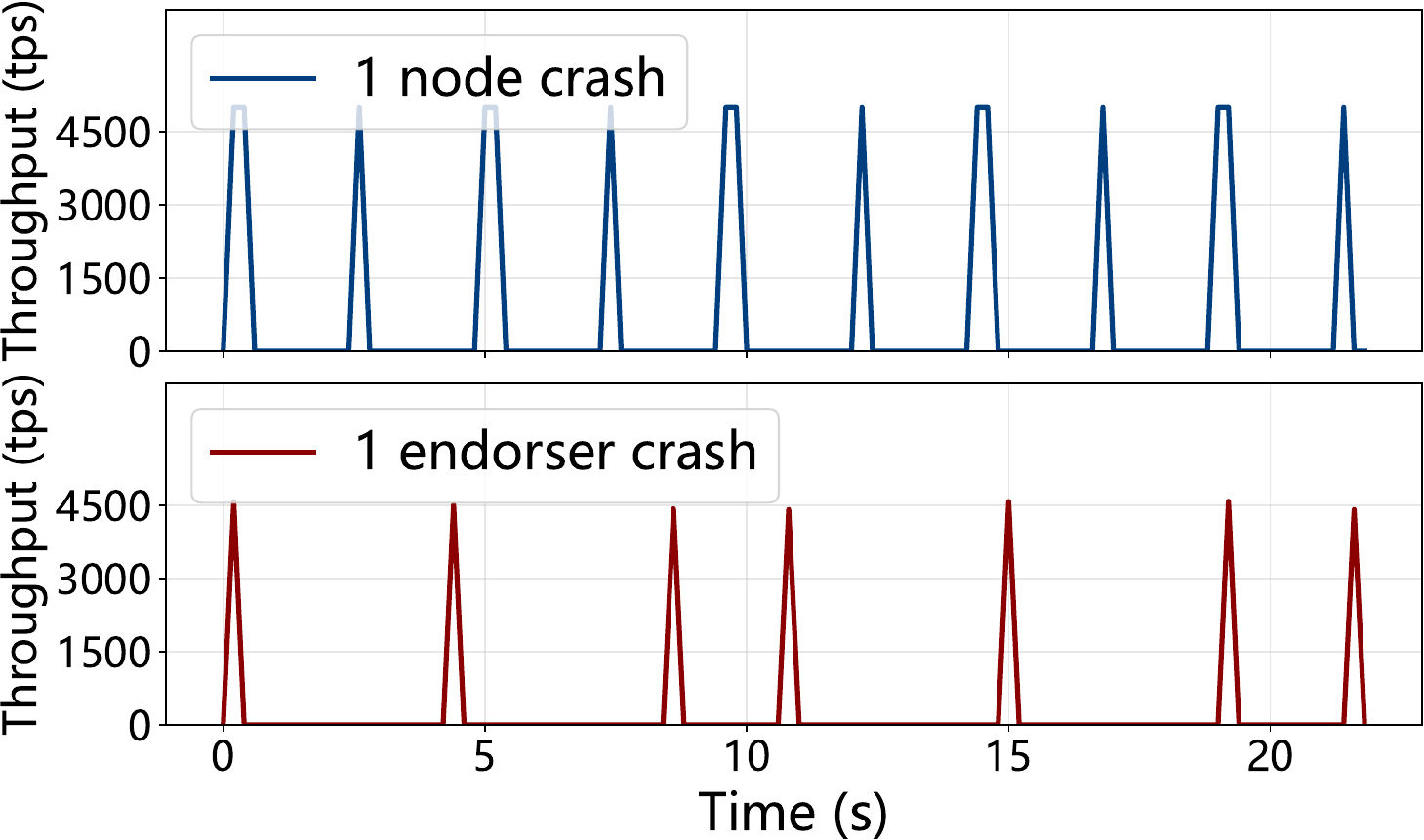}
          \caption{Throughput when one node or endorser failed ($n=4$).}
        \label{fig:timeout}
    \end{minipage}
    \hfill
    \begin{minipage}{0.32\linewidth}
       \centering
        \includegraphics[scale=0.18]{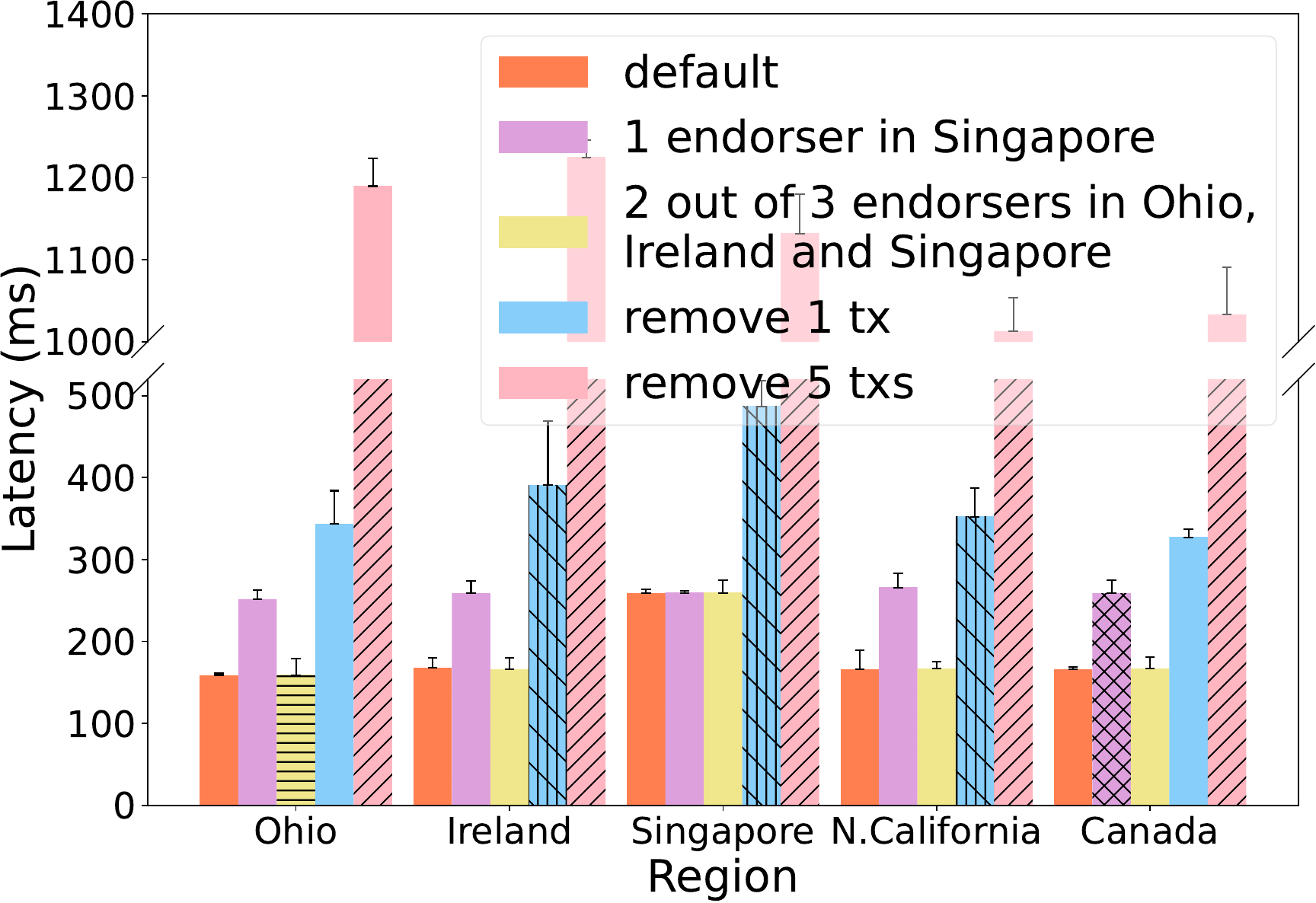}
        \caption{Average and $99^{th}$ percentile latency of \sys in WANs ($n=10$).}
        \label{fig:wan}
    \end{minipage}
    \vspace{-15pt}
\end{figure*}

\vspace{-5pt} 
\subsection{Transaction removal}

As endorsers can explicitly oppose a transaction via \prevote{} messages (i.e., rapid removal), we first evaluate \sys under this situation with $n=4$.
We use the default endorsement policy (i.e., any two-thirds of nodes), so a transaction can be rapidly removed once two nodes oppose it. 
In each round, endorsers oppose one transaction, leading to its removal and the re-execution of the remaining transactions.
We further distinguish two cases: whether the re-execution optimization is enabled.
The workload is USDT transfers unless otherwise stated.
The results are depicted in Figure~\ref{fig:fastremove}.
As more transactions are removed, \sys requires additional rounds for re-execution and re-endorsement, badly impacting its peak throughput.
This overhead includes both additional message exchanges and repeated execution.
With the re-execution optimization enabled, only causally dependent transactions are re-executed, resulting in a noticeable improvement over the unoptimized case.
We further examine the re-execution overhead under different workloads. 
The overhead is more pronounced when all transactions access the same accounts (i.e., the all-conflict case) and decreases when transactions are conflict-free.
In the second case, the overhead stems primarily from additional messaging.

We then investigate how \sys behaves when a single node or endorser crashes.
Upon the expiration of $tmr_{V,r}$, the remaining nodes suggest that all transactions requiring the specific endorsement be removed.
The evolution of the throughput in a four-node setup is shown in Figure~\ref{fig:timeout}.
Even when the crashed node is not the endorser, \sys no longer sustains continuous throughput.
This is because \sys and Tendermint proactively rotate the primary at each block height, which is essential for fair block generation.
Whenever the crashed node is selected as the primary, it is unable to drive consensus, causing the protocol to stall until a timeout occurs.
With a failed endorser, the throughput ``valleys'' are sometimes wider than those with a correct endorser.
This is because, once the endorser becomes the primary, \sys requires one round to replace it and another to remove pending transactions, with both rounds triggered by timeout events.
The results highlight the need to make endorsers themselves more decentralized so as to ensure continuous service.

\vspace{-5pt}
\subsection{Latency in WANs}

We next evaluate both protocols over WAN deployments spanning five geographic regions.
Each region hosts two nodes (i.e., $n=10$) and the workload is USDT transfers.
In this setting, neither protocol reaches saturation.
Figure~\ref{fig:wan} depicts the results.
Because \eov's latencies exceed two seconds, we omit them from the figure.
The latency gap between \sys and \eov widens further, since \eov incurs larger delays in WANs when retrying aborted transactions.
The impact of endorsement policies is also exacerbated.
As the single endorser is located in Singapore,  other nodes must wait its \prevote{} message even if they are closer and constitute a quorum.
When endorsers are distributed more widely (2 out of 3), the latency impact diminishes because endorsers are closer to other nodes.
We further measure the latencies of \sys when some transactions are rapidly removed.
The latencies increase roughly in proportion to the number of rounds required for transaction removal, amplifying the latency overhead caused by additional message exchanges.

\vspace{-5pt}
\section{Related work}
\vspace{-5pt}
\subsection{Enhancements of Fabric}

Numerous works have investigated Fabric's performance and scalability.
FastFabric~\cite{fastFabric} and SmartFabric~\cite{sparsepeer} accelerate Fabric by removing redundant validation work and parallelizing transaction processing.
FabricCRDT~\cite{fabricCRDT} incorporates conflict-free replicated datatypes~\cite{crdt} to improve parallelism and reduce abort rates.
While these works are largely orthogonal to \sys, certain ideas can be further leveraged.

Chacko \emph{et~al.}~\cite{whyfail} 
found that transaction failures in Fabric can exceed 40\% in realistic scenarios due to concurrency conflicts.
Fabric++~\cite{fabric++} improves Fabric's throughput by reducing unnecessary MVCC-based transaction aborts.
Ruan \emph{et~al.}~\cite{frabirsharp} 
demonstrated that Fabric and Fabric++ enforce a stronger-than-necessary serializability constraint, causing excessive transaction aborts.
They further proposed FabricSharp and FastFabricSharp, which introduce a fine-grained transaction-reordering technique that avoids unnecessary aborts.
ConChain~\cite{ConChain} combines a dependency manager with parallel ordering to minimize conflicts.
Kaul \emph{et~al.}~\cite{aware} proposed a dependency-aware execution model that detects transaction conflicts early.
These approaches can mitigate the problem caused by unnecessary aborts, whereas \sys completely eliminates conflict-induced aborts.

XOX Fabric~\cite{xox} and HTFabric~\cite{htfabric} employ a re-execution mechanism to avoid repeated transaction aborts in high-contention scenarios, dramatically reducing re-execution latency.
However, by doing so, re-executed transactions no longer meet the requirement of being endorsed by specific peers, whereas in \sys the execution result of every committed transaction is properly endorsed.

\vspace{-10pt}
\subsection{Improvement on (classical) BFT protocols}

\noindent\textbf{Parallel execution.}
Bidl~\cite{bidl}, as a permissioned blockchain designed for datacenter networks, achieves low latency and high throughput by aggressively decoupling execution from consensus and parallelizing contract execution.
Block-STM~\cite{blockstm} uses optimistic Software Transactional Memory (STM) combined with multi-versioned data structures to exploit parallelism while preserving a predetermined transaction order.
Vegeta~\cite{vegeta} allows multiple proposers to simulate transactions before consensus, enabling all nodes to replay them efficiently.
ParallelEVM~\cite{parallelevm} introduces an operation-level concurrency control algorithm for Ethereum, allowing only the specific operations involved in a conflict to be re-executed, rather than aborting entire transactions.
These approaches are amenable to further integration with \sys.
 
\noindent\textbf{Scalability.}
Blockchain or BFT systems primarily rely on sharding~\cite{elastico,dynamicsharding,reducecrossshard,Basil} or layer-2 solutions~\cite{arbitrum,formalzkrollup,securesequencer} to scale.
For instance, Basil~\cite{Basil} is a leaderless transactional key-value store that prevents redundant (cross-shard) coordination on concurrency control, replication and two-phase commit.
Arbitrum~\cite{arbitrum} enables layer-2 smart contracts by executing them off-chain with only minimal on-chain verification.
When paired with a scaling solution, \sys can further mitigate delays caused by unendorsed transactions, since irrelevant transactions reside in different shards or off-chain contracts.
We leave this promising direction as future work.

\vspace{-5pt}
\subsection{Regulation on blockchain}

Bodo and Filippi~\cite{trustincontext} argued that while DeFi presents itself as ``trustless'', it still relies on multiple social and institutional layers of trust such as developers, auditors and governance participants~\cite{blockgovernance}.
Liu \emph{et al.}~\cite{evasion} analyzed 957 days of Ethereum activity and showed that, even with addresses sanctions, malicious actors can still evade enforcement at scale.
Flexible endorsement (as in Fabric and \sys) offers a proactive way to prevent potential misuse.
For instance, the authors of~\cite{offramps} suggested that blockchain provenance can be used to block illicit funds at ``off-ramps'' during crypto-to-fiat conversions, where authorities can enforce regulations with the help of endorsement requirements.
Flexible endorsement can also be regarded as complementary to on-chain governance~\cite{blockgovernance}, where the final result must be endorsed by specific players.

As existing anti-money laundering (AML) tools struggle to keep pace with blockchain's evolving laundering techniques, machine-learning-based approaches~\cite{mlmoneylauderydetect,combatingmoneylaundering,deepdiveaml} are increasingly used to detect and identify illicit transactions.
In this scenario, detection programs can run in the background, and once a suspicious address or account is identified, \sys can block any subsequent transactions involving it.

\vspace{-5pt}
\section{Conclusion}
\vspace{-5pt}

To integrate regulation into order-execute blockchains in a practical manner, we proposed \sys, a conflict-free BFT protocol with built-in support for customizable endorsement policies.
\sys in the best case incurs no additional overhead compared to Tendermint.
We believe our approach offers a promising path for the adoption of decentralized systems in traditional financial services, by combining the predominant blockchain paradigm with real-time proactive regulation.
More broadly, our approach enriches the functionality and semantics of protocols designed to solve the original consensus problem.
Other application requirements that demand iterative solutions can be accommodated in a similar manner.
\section*{Acknowledgments}

This work was supported in part by the Shanghai Action Plan for Science, Technology and Innovation (grant no. 24BC3201300) and the National Natural Science Foundation of China (grant no. 62372293).


\bibliographystyle{plain}
\bibliography{bibliography}

\appendix

\section{Correctness proof}
\label{sec:correctness}

\subsection{Safety}

\begin{theorem}{(Safety)}
If node $i$ commits proposal $v$ in round $r$ and node $j$ commits proposal $v'$ in round $r'\geq r$, then $v=v'$ and $v$ is properly endorsed. 
\end{theorem}
\begin{proof}
The safety property of \sys is a direct consequence of Tendermint's locking mechanism,
as a proposal can be locked only after \emph{its endorsements have been collected}.
The locking mechanism also ensures that proposal $v$ can be committed only when $v$ is locked by a quorum of $2f+1$ nodes, thereby making $v$ permanent across subsequent rounds.
If another node $j$ commits $v'$ in a later round $r'\geq r$, a quorum of nodes must approve $v'$ in some later round between $r$ and $r'$.
As only $v$ satisfies the unlocking condition, it holds that $v=v'$.
\end{proof}

\subsection{Liveness}

We first give the proofs of Invariant~\ref{inv:reviewed} and Lemma~\ref{lem:reviewed}, which together lead to the proof of Theorem~\ref{the:termination}.
We assume \textbf{no proposal was properly endorsed before round $r$}, otherwise we can skip Invariant~\ref{inv:reviewed} and Lemma~\ref{lem:reviewed}.

\ReviewInvariant*
\begin{proof}
If transaction $tx$ in $v$ is not properly endorsed by the end of the \prevote{} phase in round $r$, 
or certified for rapid removal in the \prevote{} phase in round $r$,
correct nodes will suggest removing $tx$ via \precommit{} messages.
Thus, at least $2f+1$ nodes will suggest its removal, making $tx$ removable in subsequent rounds.
If no such removable transaction exists, then for every transaction in $v$, at least one correct node must have observed its endorsement proof during the \prevote{} phase.
Due to the gossip communication model, every correct node will eventually receive the proof of all transactions, making $v$ properly endorsed.
\end{proof}

\begin{lemma}
\label{lem:reviewed}
Assume every correct node enters round $r$ after GST, and the proposer of round $r+1$ is a correct node and proposes $v$.
Then, $v$ is \reviewed in round $r+1$.
\end{lemma}
\begin{proof}
We first focus on proving that $v$ can be approved by every correct node in round $r+1$; that is, every correct node will send a \prevote{} message in round $r+1$ for $v$. 
We discuss two cases.

\noindent\textbf{Case I.}
If some proposal was \reviewed before round $r+1$, then
the proposer of round $r+1$, denoted as node $q$, should reference an \reviewed proposal $v'$ and its corresponding round $r'$.
We further distinguish two sub-cases:
Proposer $q$ enters round $r+1$ (1) when $tmr_{C,r}$ expires or (2) when $q$ receives $2f+1$ \precommit{} messages of round $r$ for some value.
In the former sub-case, no proposal should be \reviewed in round $r$, because after GST, $q$ must have received every non-nil \precommit{} message sent by correct nodes before entering round $r+1$,
and because if no correct node has locked on any value (which is true because we assume no proposal is yet properly endorsed), then once any correct node sends a non-nil \precommit{} message in round $r$, every correct node will send the same \precommit{} before its $tmr_{V,r}$ expires.
Since some \reviewed proposal exists before round $r$, by $tmr_{C,r}$, $q$ should have received an \reviewed proposal before $tmr_{C,r}$ expires. 
In the latter sub-case, round $r$ and its \reviewed value can be referenced by $q$.
Thus, $q$ should reference an \reviewed value, say in round $r'$, upon proposing $v$ in round $r+1$.

We argue that $v$ can be approved by every correct node, provided that the node has not locked on any value (which is true because we assume no proposal is yet properly endorsed) and received all \precommit{} messages of round $r'$ propagated by $q$.

\noindent\textbf{Case II.}
If no proposal was \reviewed before round $r+1$,
then the $refRound_*$ of every correct node should be $-1$.
Thus, $v$ can be approved by every correct node.

After GST, every correct node receives $v$ (and \precommit{} messages of the reference round) before $tmr_{P,r+1}$ expires, thereby approving $v$ and further proceeding to the \precommit{} phase. 
Finally, $v$ is \reviewed in round $r+1$.
\end{proof}

\begin{theorem}{(Termination)}
\label{the:termination}
After GST, some proposal can be committed eventually. 
\end{theorem}
\begin{proof}
Given Invariant~\ref{inv:reviewed} and Lemma~\ref{lem:reviewed}, and since each correct proposer references an \reviewed proposal with the minimum number of transactions and removes at least one transaction from it, eventually, there will exist a proposal that is properly endorsed (maybe an empty one).
From that point on, each correct proposer only proposes an endorsed proposal, which will be accepted by every correct node (unless it is locked on another endorsed value).
In general, the protocol enters a mode that is equivalent to the standard Tendermint.
Thus, by the liveness property of Tendermint, every correct node eventually commits a proposal.
\end{proof}

We then give the proof of Invariant~\ref{inv:censorship}, which is critical to censorship resistance.
\censorship*
\begin{proof}
Assume node $q$ is the proposer in round $r$.
We first argue that after GST, if $q$ proposed both $v$ and $v'$ in round $r$, causing two correct nodes to send \prevote{} messages for $v$ and $v'$ ($v\neq v'$), respectively, then no correct node suggests removal in round $r$.

Let $t$ denote the earliest time at which a correct node, say node $j$, sent a \prevote{} message.
Without loss of generality, further assume correct node $j$ sent a \prevote{} message for $v$ and correct node $k$ sent a \prevote{} message for $v'$.
Node $j$ also propagated \propose{} $v$ before time $t$ due to the gossip communication model.
Thus, after GST, node $k$ must receive \propose{} $v'$ before $t+\Delta$, otherwise it should first receive $v$.
Node $k$ thus propagated $v'$ before $t+\Delta$.
Every correct node received $v'$ before time $t+2\Delta$ and $v$ before time $t+\Delta$, detecting conflicts before $t+2\Delta$.
Since transaction $tx$ cannot be rapidly removed, each correct node enters the \precommit{} phase either upon receiving sufficient endorsements and $2f+1$ \prevote{} messages for $v$, or when $tmr_{V,r}$ expires.
As node $j$ is the first correct node to send \prevote{} (at time $t$), no correct node will enter the \precommit{} phase before $t+2\Delta$, except those that have received enough endorsements for every transaction in $v$.
Obviously, they will not suggest any removal in the \precommit{} phase.
Thus, every correct node either sends a $nil$ \precommit{} upon receiving two conflicting proposals or sends a \precommit{} with no removal.

Assume $v$ is \reviewed.
As the above-mentioned case cannot lead to any removal, we further assume each correct node either sent a \prevote{} $v$ or a \prevote{} $nil$, but not a \prevote{} $v'\neq v$.
After GST, every correct node received $v$ before $t+\Delta$.
As mentioned in \S\ref{sec:threat}, endorsers still send a (new) \prevote{} message regardless of whether they have already sent \prevote{} $nil$ , thus other nodes receive their \prevote{} messages before time $t+2\Delta$.
That means, every correct node will receive all endorsements sent by correct nodes before entering the \precommit{} phase.
Thus, no correct node will suggest any removal in the \precommit{} phase.
\end{proof}



\begin{theorem}{(Liveness)}
After GST, any transaction $tx$ issued by a correct client will be committed, provided that a sufficient number of endorsers endorse it and it cannot be rapidly removed.
\end{theorem}
\begin{proof}
As long as $tx$ is received by a correct node and added into its mempool, the node will eventually become a proposer and propose a block containing $tx$. 
After GST, and by Invariant~\ref{inv:censorship}, $tx$ will not be removed in any round.
By Theorem~\ref{the:termination}, $tx$ will be committed.
\end{proof}

\section{Pseudocode of \sys}
\label{sec:pseudocode}

We present the pseudocode of \sys in Algorithms~\ref{alg:sys},~\ref{alg:func} and~\ref{alg:fast}, which respectively describes the events, functions, and rapid-removal mechanism.
We highlight the modifications to Tendermint in \algemph{grey}.

\begin{algorithm*}[!ht]
  \caption{\sys code for node $p$ (Part I: events).}\label{alg:sys}
	\scriptsize
	\textbf{Init}: $h_p\leftarrow 0, round_p\leftarrow 0,step_p\in\{\propose,\prevote,\precommit\},decision_p[]\leftarrow nil,lockedValue_p\leftarrow nil,lockedRound_p\leftarrow -1, validValue_p\leftarrow nil,$ $validRound_p\leftarrow -1$, \algemph{$\refround\leftarrow -1,\txs\leftarrow nil$}
  \begin{algorithmic}[1]
  \State \textbf{upon} start \textbf{do} $startRound(0)$
  \State \textbf{upon} \msg{\proposal}{h_p,round_p,v,-1,\algemph{$rr$}} from \texttt{proposer}$(h_p,round_p)$ \textbf{while} $step_p=\propose$\algemph{$\wedge((rr=-1\wedge \refround=-1)\vee VerifyReference(rr,v))$} \textbf{do}
  \State \myindent \textbf{if} $valid(v)\wedge(lockedRound_p=-1\vee lockedValue_p=v)$ \textbf{then}
  \State \myindent\myindent \textbf{execute} transactions in $v$, \algemph{\textbf{verify} the results and \textbf{insert} the endorsed transactions into $D$}
  \State \myindent\myindent \textbf{broadcast} \msg{\prevote}{h_p,round_p,hash(v),\algemph{$D$}}
  \State \myindent \textbf{else}
  \State \myindent\myindent \textbf{broadcast} \msg{\prevote}{h_p,round_p,nil,\algemph{$nil$}}
  \State \myindent $step_p\leftarrow \prevote$
  \newline
  \State \textbf{upon} \msg{\proposal}{h_p,round_p,v,vr,\algemph{$rr$}} from \texttt{proposer}$(h_p,round_p)$ \textbf{AND} $2f+1$ \msg{prevote}{h_p,vr,hash(v),\algemph{$*$}} \textbf{while} $step_p=\propose\wedge vr\geq 0\wedge vr<round_p\wedge$\algemph{$VerifyEndorsement(rr,v)$} \textbf{do} \Comment when $vr\neq -1$, no re-execution or re-endorsement is needed
  \State \myindent \textbf{if} $valid(v)\wedge(lockedRound_p\leq vr\vee lockedValue_p=v)$  \textbf{then}
  \State \myindent\myindent \textbf{broadcast} \msg{\prevote}{h_p,round_p,hash(v),\algemph{$nil$}}
  \State \myindent \textbf{else}
  \State \myindent\myindent \textbf{broadcast} \msg{\prevote}{h_p,round_p,nil,\algemph{$nil$}}
  \State \myindent $step_p\leftarrow \prevote$
  \newline
  \State \textbf{upon} $2f+1$ \msg{\prevote}{h_p,round_p,*,\algemph{$*$}} \textbf{while} $step_p=\prevote$ for the first time \textbf{do}
  \State \myindent \textbf{schedule} $OnTimeoutPrevote(h_p,round_p)$ to be executed \textbf{after} $timeoutPrevote(round_p)$
  \newline
  \State \textbf{upon} \msg{\proposal}{h_p,round_p,v,*,\algemph{$rr$}} from \texttt{proposer}$(h_p,round_p)$ \textbf{AND} \algemph{at least} $2f+1$ \msg{\prevote}{h_p,round_p,hash(v),\algemph{$*$}} \textbf{while} $valid(v)\wedge step_p=\prevote$\algemph{$\wedge (VerifyEndorsement(rr,v)\vee VerifyEndorsement(round_p,v))$} for the first time \textbf{do}
  \State \myindent $lockedValue_p\leftarrow v$
  \State \myindent $lockedRound_p\leftarrow round_p$
  \State \myindent \textbf{broadcast} \msg{\precommit}{h_p,round_p,hash(v),\algemph{$nil$}}
  \State \myindent $step_p\leftarrow \precommit$
  \newline
  \State \textbf{upon} \msg{\proposal}{h_p,r,v,*,\algemph{$rr$}} from \texttt{proposer}$(h_p,r)$ \textbf{AND} \algemph{at least} $2f+1$ \msg{\prevote}{h_p,r,hash(v),\algemph{$*$}} \textbf{while} $valid(v)\wedge r>validRound_p$\algemph{$\wedge (VerifyEndorsement(rr,v)\vee VerifyEndorsement(r,v))$} \textbf{do}
  \State \myindent $validValue_p\leftarrow v$
  \State \myindent $validRound_p\leftarrow r$
  \State \myindent \algemph{\textbf{if} $VerifyEndorsement(r,v)$ \textbf{then}} \Comment $\refround$ is updated to the round in which $validValue_p$ is properly endorsed
  \State \myindent\myindent \algemph{$\refround\leftarrow r$}
  \State \myindent \algemph{\textbf{else}}
  \State \myindent\myindent \algemph{$\refround\leftarrow rr$}
  \newline
  \State \textbf{upon} $2f+1$ \msg{\prevote}{h_p,round_p,nil,\algemph{$nil$}} while $step_p=\prevote$ \textbf{do}
  \State \myindent \textbf{broadcast} \msg{\precommit}{h_p,round_p,nil,\algemph{$nil$}}
  \State \myindent $step_p\leftarrow \precommit$
  \newline
  \State \textbf{upon}  $2f+1$ \msg{\precommit}{h_p,round_p,*,\algemph{$*$}} for the first time \textbf{do}
  \State \myindent \textbf{schedule} $OnTimeoutPrecommit(h_p,round_p)$ to be executed \textbf{after} $timeoutPrecommit(round_p)$
  \newline
  \State \textbf{upon} \msg{\proposal}{h_p,r,v,*,\algemph{$*$}} from \texttt{proposer}$(h_p,r)$ \textbf{AND} $2f+1$ \msg{\precommit}{h_p,r,hash(v),\algemph{$nil$}} \textbf{while} $decision_p[h_p]=nil$ \textbf{do}
  \State \myindent \textbf{if} $valid(v)$ \textbf{then}
  \State \myindent\myindent $decision[h_p]\leftarrow v$
  \State \myindent\myindent $h_p\leftarrow h_p+1$
  \State \myindent\myindent reset $lockedRound_p,lockedValue_p,validRound_p,validValue_p$\algemph{$,\refround,\txs$} to initial values and empty message log
  \State \myindent\myindent $StartRound(0)$
  \newline
  \State \textbf{upon} $f+1$ \msg{\ensuremath{*}}{h_p,round,*,*} with $round>round_p$ \textbf{do}
  \State \myindent $StartRound(round)$
  \State \algemph{\textbf{upon} \msg{\proposal}{h_p,r,v,-1,*} from \texttt{proposer}$(h_p,r)$ \textbf{AND} $2f+1$ \msg{\precommit}{h_p,r,hash(v),\algemph{$*$}} \textbf{do}} \Comment $v$ is \reviewed in round $r$
  \State \myindent \algemph{\textbf{if} $\refround=-1\vee(validRound_p=-1\wedge length(v)<length(\txs))$ \textbf{then}}
  \State \myindent\myindent \algemph{$\refround\leftarrow r$}
  \State \myindent\myindent \algemph{$\txs\leftarrow v$}
  \State \myindent \algemph{\textbf{if} $round_p=r$ \textbf{then}}
  \State \myindent\myindent \algemph{$StartRound(round_p+1)$}
  \State \algemph{\textbf{upon} \msg{\proposal}{h_p,round_p,v,*,*} \textbf{AND} \msg{\proposal}{h_p,round_p,v',*,*}  from \texttt{proposer}$(h_p,r)$  \textbf{while} $v\neq v'\wedge step_p < \precommit$ \textbf{do}}
  \State \myindent \algemph{\textbf{broadcast} \msg{\precommit}{h_p,round_p,nil,nil}}
\algstore{sys}
\end{algorithmic}
\end{algorithm*}

\begin{algorithm*}
\scriptsize
\caption{\sys code for node $p$ (Part II: functions).}\label{alg:func}
\begin{algorithmic}[1]
\algrestore{sys}
  \State \textbf{Function} $startRound(round):$
  \State \myindent $round_p\leftarrow round$
  \State \myindent $step_p\leftarrow \propose$
  \State \myindent \textbf{if} $\texttt{proposer}(h_p,round_p)=p$ \textbf{then}
  \State \myindent\myindent \textbf{if} $validValue_p\neq nil$ \textbf{then}
  \State \myindent\myindent\myindent $proposal\leftarrow validValue_p$
  \State \myindent\myindent \textbf{else}
  \State \myindent\myindent\myindent \algemph{\textbf{if} $\txs\neq nil$ \textbf{then}}
  \State \myindent\myindent\myindent\myindent \algemph{$proposal\leftarrow Extract(\refround,\txs)$}
  \State \myindent\myindent\myindent \algemph{\textbf{else}}
  \State \myindent\myindent\myindent\myindent $proposal\leftarrow getValue()$\Comment get a new value
  \State \myindent\myindent \textbf{broadcast} \msg{\proposal}{h_p,round_p,proposal,validRound_p,\algemph{$\refround$}}
  \State \myindent \textbf{else}
  \State \myindent\myindent \textbf{schedule} $OnTimeoutPropose(h_p,round_p)$ to be executed \textbf{after} $timeoutPropose(round_p)$
  \newline  
  \State \textbf{Function} $OnTimeoutPropose(height,round):$
  \State \myindent \textbf{if} $height=h_p\wedge round=round_p\wedge step_p=\propose$ \textbf{then}
  \State \myindent\myindent \textbf{broadcast} \msg{\prevote}{h_p,round_p,nil,\algemph{$nil$}}
  \State \myindent\myindent $step_p\leftarrow \prevote$
  \newline
  \State \textbf{Function} $OnTimeoutPrevote(height,round):$
  \State \myindent \textbf{if} $height=h_p\wedge round=round_p\wedge step_p=\prevote$ \textbf{then}
  \State \myindent\myindent \algemph{\textbf{if} \msg{\proposal}{h_p,round_p,v,-1,rr} from \texttt{proposer}$(h_p,round_p)$ \textbf{AND} $2f+1$ \msg{\prevote}{h_p,round_p,hash(v),\algemph{$*$}} \textbf{then}}
  \State \myindent\myindent\myindent \algemph{\textbf{broadcast} \msg{\precommit}{h_p,round_p,hash(v),GetExcludedTx(r,v)}}
  \State \myindent\myindent \textbf{else}
  \State \myindent\myindent\myindent \textbf{broadcast} \msg{\precommit}{h_p,round_p,nil,\algemph{$nil$}}
  \State \myindent\myindent $step_p\leftarrow \precommit$
  \newline
  \State \textbf{Function} $OnTimeoutPrecommit(height,round):$
  \State \myindent \textbf{if} $height=h_p\wedge round=round_p$ \textbf{then}
  \State \myindent\myindent $StartRound(round_p+1)$
  \newline
  \State \algemph{\textbf{Function} $VerifyReference(r,v):$}
  \State \myindent \algemph{\textbf{return} $\exists v'$ and $2f+1$ \msg{\precommit}{h_p,r,hash(v'),*} such that $v\subset v'$ \textbf{and} for $tx\in v'-v$: at least $f+1$ \msg{\precommit}{h_p,r,hash(v'),*} exclude $tx$}
  \newline
  \State \algemph{\textbf{Function} $VerifyEndorsement(r,v):$}
  \State \myindent \algemph{\textbf{return} $\exists$ a set of \msg{\prevote}{h_p,r,hash(v),*} that makes $\forall tx\in v$ properly endorsed}
  \newline
  \State \algemph{\textbf{Function} $GetExcludedTx(r,v):$}
  \State \myindent \algemph{\textbf{return} $\{\forall tx\in v: tx \textnormal{ is not properly endorsed by } \msg{\prevote}{h_p,r,hash(v),*} \}$}
  \newline
  \State \algemph{\textbf{Function} $Extract(r,v):$}
  \State \myindent \algemph{\textbf{return} $v-\{\textnormal{the first } tx\in v: \textnormal{at least } f+1 \textnormal{ } \msg{\precommit}{h_p,r,hash(v),*} \textnormal{ exclude } tx\}$}
\algstore{sys}
\end{algorithmic}
\end{algorithm*}


\begin{algorithm*}
\scriptsize
\caption{\sys code for node $p$ (Part III: rapid removal).}\label{alg:fast}
\begin{algorithmic}[1]
\algrestore{sys}
  \State \textbf{upon} \msg{\proposal}{h_p,round_p,v,-1,rr} from \texttt{proposer}$(h_p,round_p)$ \textbf{AND} at least $2f+1$ \msg{\prevote}{h_p,round_p,hash(v),*} \textbf{while} $valid(v)\wedge step_p=\prevote\wedge \forall tx\in v: tx \text{ is properly endorsed or can be removed}$ \textbf{do}
  \State \myindent \textbf{broadcast} \msg{\precommit}{h_p,round_p,hash(v),GetExcludedTx(r,v)}
  \State \myindent $step_p\leftarrow \precommit$
  \newline
    \State \textbf{upon} \msg{\proposal}{h_p,r,v,*,$*$} from \texttt{proposer}$(h_p,round_p)$ \textbf{AND} $2f+1$ \msg{\precommit}{h_p,r,hash(v),$*$} \textbf{while} $decision_p[h_p]=nil$ \textbf{do}
\State \myindent \textbf{if} $round_p=r$ \textbf{then}  
\State \myindent\myindent $StartRound(round_p+1)$
\end{algorithmic}
\end{algorithm*}

\end{document}